\documentclass[11pt]{article}

\usepackage[margin=1in]{geometry}
\usepackage{mathpazo}
\usepackage[T1]{fontenc}
\usepackage{amsmath,amssymb,amsthm}
\usepackage{mathtools}
\usepackage{bm}
\usepackage{minitoc}
\usepackage[dvipsnames]{xcolor}
\usepackage[natbib=true,sorting=ynt,sortcites,style=alphabetic,maxbibnames=99,maxcitenames=2,minalphanames=3]{biblatex}
\usepackage{enumitem}
\usepackage{color-edits}
\usepackage{booktabs}
\addauthor[SC]{sc}{purple}
\addauthor[AI]{ai}{purple}

\newcommand{\pcf}{p_\text{CF}}
\usepackage{amsmath}
\usepackage{amsthm}
\usepackage{amssymb}
\usepackage{mathtools}
\usepackage{mathstuff}
\usepackage{my_notation}
\usepackage{todonotes}
\usepackage{lipsum}

\usepackage[labelfont=bf]{caption} %

\renewcommand{\star}{\ensuremath{*}}

\title{User Strategization and Trustworthy Algorithms}
\author{Sarah H. Cen, Andrew Ilyas, and Aleksander  M\k{a}dry \\
MIT \\
\texttt{\{shcen,ailyas,madry\}@mit.edu}
}
\date{}

\begin{document}
\doparttoc %
\faketableofcontents %
\maketitle

\begin{abstract}
	Many human-facing algorithms---including those that power recommender systems or hiring decision 
	tools---are trained on data provided by their users. 
	The developers of these algorithms commonly adopt the assumption that the data generating process is {\em exogenous}: 
	that is, 
	how a user reacts to a given prompt (e.g., a recommendation or hiring suggestion) depends on the prompt and \emph{not} on the algorithm that generated it. 
	For example, the assumption that a person's behavior follows a ground-truth distribution is an exogeneity assumption. 
	In practice, when algorithms interact with humans, this assumption rarely holds because users can be {\em strategic}. 
	Recent studies document, for example,
	TikTok users changing their scrolling behavior after learning that 
	TikTok uses it to curate their feed, and Uber
	drivers changing how they accept and cancel rides in response to 
	changes in Uber's algorithm.

	Our work studies the implications of this strategic behavior by modeling the interactions between a user and
	their data-driven platform as a repeated, two-player game. 
	We first find that user strategization 
	can actually help platforms in the short term. 
	We then show that it corrupts platforms'
	data and ultimately hurts their ability to make counterfactual decisions.
	We connect this phenomenon to user trust,
	and show that designing trustworthy algorithms can go hand in hand with
	accurate estimation. 
	Finally, we provide a formalization of
	trustworthiness that inspires potential interventions.
\end{abstract}

\section{Introduction}

In the age of personalization,
data-driven platforms have become increasingly reliant on data provided by their users. 
Platforms like Facebook, Amazon, and Uber tailor their services to each user based on the user's interaction history. 
Even data-driven tools in medicine and hiring utilize previous interactions in order to fine-tune and improve their results. 

Traditionally, platforms make a key assumption when processing user data, 
namely that the data is generated 
\emph{exogenously}, i.e., how a user behaves depends on the user but \emph{not} on the platform’s algorithm.
When a user's behavior is exogenous,
an action---say, skipping a video on YouTube---is a reflection of the user's preference for that video and {not} of YouTube's  algorithm: 
the user would skip the video regardless of the algorithm that generated it. 
YouTube can therefore attribute the skip to the user's content preferences and recommend future content accordingly.

In practice, however, users are not blind to how their platforms operate. 
In fact, many  users adapt their behavior---or  \emph{strategize}---based on their perception of how the platform works. 
For instance, a Facebook user might not click on a post not because they find it uninteresting, but because they believe the algorithm will over-recommend similar content in the future if they do click. Or an Uber driver may cancel low-paying rides that they would normally accept because they have learned that Uber’s algorithm does not penalize them for excessive cancelations \cite{wireduber}. 
In fact, \citet{empirical} recently document user strategization in a lab experiment,
discovering that over half of users consciously strategize on online platforms.

\paragraph{Contributions.} 
 In this paper, 
we study the implications of 
user strategization and establish that it can be 
both
helpful and harmful to 
platforms. 
Using a game-theoretic model of strategization (which \cite{empirical} find to be empirically predictive of strategization), 
we first demonstrate that strategization can provide the platform with cues that
the platform would otherwise miss. 
We then show, however, that the platform's data is no longer generalizable when users strategize---data gathered under one algorithm
cannot be used to make reliable predictions in other contexts. 
This outcome is undesirable for platforms that may wish to use the data for other purposes, 
e.g., to train alternate algorithms,
begging the question: 
What should platforms do when users strategize?

To work around user strategization,
some platforms attempt to gather data that is more likely to satisfy the exogeneity assumption. 
For instance, 
Facebook may notice that users strategize by selectively clicking the “like” button and begin tracking how quickly users scroll down their feeds, using the dwell time on each post as an ``exogenous'' signal of user interest. 
However, this strategy can backfire. 
Upon learning that Facebook is tracking their dwell time, users may begin strategizing how they scroll, violating the exogeneity assumption once again and prompting a cat-and-mouse game that only serves to erode the trust between a user and their platform. 

It is natural to ask: is this outcome inevitable? 
Is out-maneuvering users the only way for platforms to obtain high-quality data?
Not necessarily. 
In this work, 
we argue that trustworthy algorithm design plays a key role. When an algorithm is trustworthy,
users do not feel compelled to strategize---they provide data that, for all intents and purposes, \emph{looks exogenous}. 
We formalize this intuition, calling a platform $\kappa$-\emph{trustworthy} when their choice of algorithm does not induce strategization and ensures users receive at least $\kappa$ utility on average. Intuitively, a platform is trustworthy when a user believes that the platform looks out for their interests along the same axis that would induce the user to strategize. We connect this formalization to the existing literature on trustworthiness, comparing it to definitions that arise in political science and law. Using this framework, we show that trustworthy design can mitigate the effects of strategization on a platform, to the benefit of both the platform and user.

\subsection{Summary of contributions}

Below, we describe our contributions and findings in greater detail.
\\[-20pt]

\paragraph{A game-theoretic model of data-driven algorithms.}
We begin by modeling the interactions between a user and their
data-driven platform as a repeated, two-player game (\cref{sec:model}).
Under our model, 
there are two agents: a user and their platform. 
At each time step,
the platform puts forth a proposition (e.g., a prediction or a recommendation), 
the user responds with a behavior (e.g., agreeing with the prediction or ignoring the recommendation),
and each party receives a payoff based on their action. 
The platform's goal is to generate high-payoff propositions by 
developing a good estimate of
the user's behavioral tendencies from repeated interactions with the user. 
In our model, 
the platform ``moves first'' in that it declares how it generates propositions, 
after which the user ``moves second'' by choosing how they respond to propositions (akin to a Stackelberg game). 
Importantly, 
this model 
allows users to observe and adapt to the platform's algorithm. 
\\[-20pt]

\paragraph{Users strategize in order to optimize their long-term outcomes.}
In Section \ref{sec:user_strat},
we use this model to define what it means for a user to be strategic.
We call a user ``naive'' if they adopt a best-response strategy,
i.e., if at every time step, they behave as though they are unaware that an algorithm uses their behavior to updates its future propositions. 
In contrast, we call a user ``strategic'' if they anticipate how their next action may affect the platform's future propositions and, in turn, the user's long-term payoff. 
That is, the strategic user understands that their actions can affect the platform's long-term behavior and adapt accordingly. 

From a technical perspective, our model of strategization
departs from previous models of non-myopic agents
in that we consider users who optimize their {\em limiting payoff}
rather than a sum of their discounted future payoffs. 
We leverage 
the notion of a \emph{globally stable set} \cite{frick2020stability}
from the misspecified learning literature
to study the platform's limiting behavior. 
This approach (a) allows us to bypass the often difficult calculations that result when optimizing a discounted sum;
and (b) better aligns with our intuition about strategization,
where users think abstractly about {\em how} platforms use 
their data, without much regard for {\em when} 
the consequences will arise. 
{As a result, our model accounts for the fact that
delay-based mechanisms (i.e., algorithms that wait to incorporate new data) 
are unlikely to ``solve'' strategization even though these mechanisms would do
so in a discounted-sum formulation.}
\\[-20pt]

\paragraph{Main results: User strategization can both help and hurt the platform.}
In Section \ref{sec:main_results}, 
we  present our main results.
We first find that, when the platform's algorithm is fixed, 
user strategization can \emph{improve} the platform's payoff 
when the user's and platform's payoffs are sufficiently aligned. 
We then show that
user strategization can \emph{mislead the platform by distorting the data that the platform collects}. 
Specifically, when a user is strategic, 
the data the platform collects under one algorithm is not reflective of the user's behavior under a different algorithm (or in a different context). 
As a result, user strategization can result in unexpected phenomena---for example,
we demonstrate that expanding the platform's hypothesis class (an action that typically improves platform performance) can actually lower the platform's payoff when users are strategic. 
As expected, \emph{these difficulties disappear when the user is not strategic}, 
i.e., when the user behaves exogenously, 
which suggests that platforms should design algorithms that do not incentivize strategization.  
\\[-20pt]

\paragraph{Trustworthy design can mitigate user strategization.}
To counter strategization,
platforms are left with a few options. 
They could simply use the collected data and risk drawing incorrect inferences.
They could post-process the data, 
but ``correcting'' the data by removing the effect of strategization is extremely challenging because a user can strategize for many different reasons and along many different axes. 
Alternatively, platforms could ``correct'' for strategic behavior by gathering more data.
In Section~\ref{sec:trust_sec}, 
we discuss when and how these approaches can fall short. 

We then discuss how another approach---designing trustworthy algorithms---can improve user and platform outcomes.
Formally, we define a $\kappa$-trustworthy algorithm 
as one that (i) does not incentivize users to strategize and 
(ii) guarantees that the user's payoff is at least $\kappa \geq 0$. 
We connect this definition to existing notions of trust,
e.g., trust as ``encapsulated interest'' \cite{hardin2006trust}.
We conclude by surfacing the reasons why users strategize, 
then discussing two interventions for trustworthy design: 
offering multiple algorithms and providing feedback mechanisms. 
Operationally,
these two mechanisms allow platforms to induce behavior that looks approximately exogenous while uncovering what users find untrustworthy. 
\section{Related work}
\label{sec:rel_work}
Our work draws inspiration from extensive lines of work spanning computer science,
economics, game theory, and the social sciences. 
We summarize a few of the most related 
areas to our work below. 
\Cref{app:full_related_work} contains 
an extended account of related work.

Conceptually related to this work are recent models of user behavior on 
recommender systems that remove the traditional ``fixed preferences''
assumption (e.g., \citep{Kleinberg2022-wy,Haupt2023-vo,Cen2022-bc}).
In short, our work differs from prior work in at least three ways:
(a) we study users that {\em adapt to the algorithm that a platform uses};
(b) we focus on the effects of this adaption on platforms;
and (c) our analysis applies to any data-driven platform that involves human
interaction. We further disambiguate from these models in \cref{app:full_related_work}.
Prior work has also empirically documented user adaptation to their 
algorithms, most notably in the context of TikTok 
\citep{Siles2022-yv,Kim2023-tv,Simpson2022-vg}.
The closest to our work is subsequent 
work by \citet{empirical}, 
who adapt our model to 
directly test for user strategization in recommendation settings
using a lab experiment.

From the economics literature, our work leverages existing 
results on misspecified 
learning \citep{Fudenberg2017-rb,Fudenberg2021-tu,bohren2016informational,Bohren2021-rs,frick2020stability} 
(in particular, that of \citet{frick2020stability})
to characterize the platform's limiting behavior. 
Furthermore, our definitions of strategization and trustworthiness are closely
related to concepts in mechanism design
and is thus connected to mechanism 
design for repeated games and auctions 
\citep{Mailath2006-oy, Bergin1993-ly, Kalai1993-ze}. 
In particular, part of the definition of trustworthiness 
is an adaptation of incentive-compatibility to our setting.

From the computer science literature, our model resembles  
that of strategic classification as well as its many variants
\citep{Levanon2022-zx,Bruckner2012-fs,Hardt2016-ks,Dong2018-rm,Ghalme2021-nf,Zrnic2021-vh,Haghtalab2022-bv}.
Compared to these works, in our model {\em all} of the following are true:
(a) users have idiosyncratic utility functions that they can optimize arbitrarily,
and are not bound to small perturbations of some ``ground-truth'' features;
(b) the platform uses a pre-specified learning algorithm (Bayesian updating);
(c) users declare a behavior {\em strategy} rather than an action;
(d) users can be myopic or non-myopic in how they interact with the platform
(and our model of non-myopic users differs from that of prior work 
\citep{Haghtalab2022-bv}).
Strategic classification is a special case of performative 
prediction \citep{Perdomo2020-za}, and the latter is also 
relevant to our work---in particular, our characterization 
of platform convergence can be viewed as a performative 
stability analysis. 
We then study the effect of data gathered in a performative setting 
when used in other contexts.

\Cref{app:full_related_work} provides an extended comparison 
with the works mentioned above, and also places our results 
into context with works from multi-agent learning 
\citep{Wu2020-bl,Hadfield-Menell2016-mv,Zhang2021-on}
and Human-AI collaboration 
\citep{Okamura2020-eu, Hou2023-cb, Ezer2019-zl, Bao2021-fn}.

\section{Model}
\label{sec:model}

In this section, we present our game-theoretic model of the interactions between a user and their data-driven platform. 
At its core, the model comprises a repeated, two-player game.
At every time step, the platform generates \emph{propositions} (e.g., recommendations).
The user then responds to these propositions  with  {\em behaviors}
(e.g., chooses whether to engage with a recommendation).

A key feature of our model is that there is no ``ground-truth'' user behavior.
In particular, 
the way that the user responds to propositions
may depend on {\em how} the platform generates them.  
Our model therefore departs from earlier works that assume a
user's behaviors are drawn from a single fixed, unknown distribution.

Since
we are interested in how users adapt to their platforms,
we study the setting where the platform first declares the
strategy it will use to generate propositions. 
The user then decides on how they wish to behave, 
with full knowledge of the platform's intended strategy.\footnote{While 
full user knowledge is a strong assumption, 
the effects we discuss here will only be exacerbated
when users have an imperfect model of the platform.
We leave {studying this case} 
to future work.}

\begin{figure}[t]
	\centering
	\includegraphics[width=\textwidth]{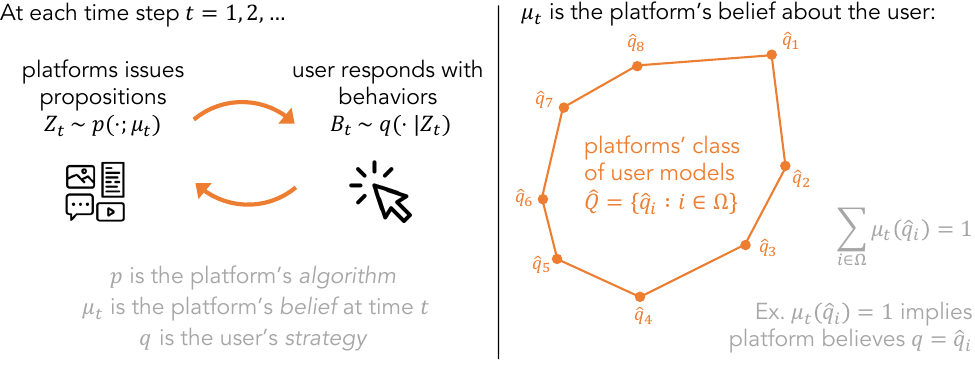}
	\caption{{\bf Illustration of the setup described in Section \ref{sec:model}}.
		{ (Left)} At each time step $t$, 
		the platform issues propositions $Z_t$, 
		and the user responds with behaviors $B_t$.
		The user's actions are determined by  
		their strategy $q: \cZ \rightarrow \Delta(\cB)$. 
		The platform's are determined by the algorithm $p$,
		the hypothesis class $\hcQ$, 
		and the platform's belief $\mu_t$ over the hypothesis class at time $t$. 
		{ (Right)} The platform's actions at time $t$ depend on its belief $\mu_t$.
		Here, $\mu_t$ is a distribution (i.e., set of weights) over $\hcQ$ such
		that $\mu_t ( \hq_i )$ denotes the probability that the platform assigns
		to the user model $q = \hq_i$ at time $t$.}
		\label{fig:setup}
\end{figure}

\subsection{Setup}
\label{sec:setup}
We model the interactions between a user and their platform as a two-player game. Formally,
{for some $d_1, d_2 > 0$},
let $\cZ \subset \bbR^{d_1}$ and $\cB \subset \bbR^{d_2}$ denote the action spaces of the platform and user, 
respectively, {where we assume that $\cB$ is finite}.
Then, at every time step $t = 0, 1, 2, \hdots$, 
\begin{enumerate}[itemsep=0pt]
	\item The platform generates \emph{propositions} $Z_t \in \cZ$.

	\item The user responds with \emph{behavior} $B_t \in \cB$,
	drawn from the {conditional} distribution $q( \cdot | Z_t)$.
	
	\item {The platform and user collect payoffs $V(Z_t, B_t)$ and $U(Z_t, B_t)$, respectively.}
\end{enumerate}
{This setup is given on the  left side of Figure \ref{fig:setup}.}
{W}e use the shorthand $q: \cZ \rightarrow \Delta(\cB)$ to denote 
{the user's {\em strategy}, i.e.,
mappings from propositions $Z$ to behavior distributions $q(\cdot|Z)$ for all $Z \in \cZ$.}
Throughout this work, $\Delta(X)$ denotes the simplex over a probability space $X$.
Furthermore, we assume that the payoffs $V$ and $U$ 
as well as the action spaces $\cZ$ and $\cB$ are {\em exogeneous},
i.e., they are pre-specified.
We assume that the payoffs $U$ and $V$ are bounded and scaled between $[0, 1]$.

\vspace*{-1em}
\paragraph{Generating propositions.}
To generate propositions, 
the platform first constructs an estimate of the user's 
strategy $q \in \cQ$. 
In particular, 
the platform believes that $q$ belongs to some set $\hcQ$,
which we assume to be finite.
$\hcQ$ can be thought of as the platform's hypothesis class,
with each
$\hq_i \in \hcQ$ 
being one of the platform's ``user models.''
The platform constructs its estimate of $q$ using
a \emph{belief} $\mu_t \in \Delta(\hcQ)$ over the hypothesis class $\hcQ$.
For instance, 
$\mu_t (\hq_i) = 1$ means that the platform believes 
$q = \hat{q}_i$ with full certainty at time step $t$.
When there is no user model that matches the user's chosen strategy (i.e., $q \not\in \hcQ$), 
we say that the platform is {\em misspecified}.

At each time step, the platform 
{uses a {\em (proposition) algorithm}}
$\smash{p: \Delta(\hcQ) \to \Delta(\cZ)}$ 
{to map its} belief $\mu_t$ to a distribution over propositions.
That is, at each time $t$, 
the platform uses its current belief $\mu_t$ to sample a proposition 
$Z_t \sim \pmu{\mu_t}{\cdot}$,
as shown on the right side of Figure~\ref{fig:setup}.\footnote{ 
One can think of $\mu_t$ as parameters of a
machine learning model trained on the data gathered until time
step $t$.}
Intuitively, the algorithm $p$ {captures} whether the platform chooses to maximize revenue, social welfare, 
or any other objective (based on its current belief $\mu_t$).

\paragraph{Bayesian updating.}
The platform incorporates its observations into its belief using Bayesian updating. 
Specifically, at each time step $t$, 
the platform updates its belief $\mu_t$ to $\mu_{t+1}$
based on the user's response $B_t$ to the platform's most recent proposition $Z_t$
 using 
Bayes' rule:
\begin{align}
	\mu_{t+1} (\hq_i) =
	\frac{\mu_{t}(\hq_i) \cdot \hq_{i} (B_t | Z_t)}{\sum_{j \in \Omega} \mu_{t}(\hq_j) \cdot \hq_{j} (B_t | Z_t) } ,
	\qquad 
	\forall\ \hq_i \in \hcQ
	.
	\label{eq:belief_update}
\end{align}
We assume that the initial belief
$\mu_0 \in\Delta(\hcQ)$ has full support, i.e., assigns a positive probability
to every possible model $\smash{\hq \in \hcQ}$.
Although we assume that the platform updates its beliefs in a Bayesian way, 
it is possible to generalize our findings to other update strategies (e.g., empirical risk minimization). 
As we discuss in Section \ref{sec:strat_user}, 
Bayesian updating allows us to leverage recent results and  precisely characterize the platform's limiting behavior as $t \rightarrow \infty$.

\paragraph{Committing to strategies.}
Thus far, 
we have {instantiated a} repeated, 
two-player game 
{between a user and their platform}. 
The platform generates propositions 
{
	by applying its {\em algorithm} $p$ to its (evolving) belief
	$\mu_t \in \Delta(\hcQ)$ in order to sample $Z_t \sim \pmu{\mu_t}{\cdot}$.
}
The user responds with behaviors $B_t \sim q(\cdot | Z_t)$. 
Note that, for a  fixed {action spaces, payoffs, and initial beliefs} 
$( \cB, \cZ, U, V , \mu_0 )$,
the user's and platform's actions are fully determined by $q$, $p$, and $\hcQ$. 
We therefore refer to $q$ as the \emph{user's strategy}
and the tuple $(p, \hcQ)$ as the {\em platform's strategy} (see Table \ref{tab:key_notation}).

In this work, we are interested in how users adapt to their platforms. We therefore study the setting in which the platform commits to a strategy $(p,
\hcQ)$ at the start of the game, after which the user chooses their strategy $q$, which may depend on $(p, \hcQ)$. In order to characterize user adaptation, we study the idealized setting in which the user has perfect knowledge of $(p,
\hcQ)$.

\subsection{Examples}

Our model captures a variety of data-driven settings. 
For instance, 
the interactions between a user and their recommender system (e.g., Netflix or
Facebook) can be viewed as a repeated, two-player game.
(Indeed, we provide a detailed recommender system example in Section
\ref{sec:example}).
Our model also captures other contexts, 
such as data-driven hiring and ride-share matching, 
as detailed below. 

\begin{example}[Hiring example]\label{ex:hiring_1}
	Suppose
	an employer (i.e., the user) uses a data-driven hiring platform to determine which job candidates to interview. 
	Then, at each time step $t$,
	$Z_t$ contains the set of candidates at time $t$ and the scores that the hiring platform assigns to the candidates. 
	$B_t$ denotes the employer's decisions (i.e., which candidates are interviewed) and the final outcomes (i.e., who is hired). 
	The employer's hiring strategy is therefore captured by $q$.
	In this context, $V(Z_t, B_t)$ is the payoff that the platform receives based on the scored candidate $Z_t$ and the employer's decisions $B_t$, 
	e.g., it could reflect a pay-per-success scheme, scaling with how many top-ranked
	applicants are successfully hired. 
	$U(Z_t, B_t)$ is the employer's corresponding payoff, 
	e.g., how many people it successfully hires using the tool. 
	
	The hiring platform's goal is to learn the employer's preferences so that it can provide scores attuned to the employer's needs. 
	It does so by first assuming that the employer's decision mechanism $q$ 
	belongs to a hypothesis class $\hQ$, 
	which describes the set of employers that the platform can model. 
	As it receives more data on the employer's hiring practices, 
	it updates its belief $\mu_t$ about the employer's preferences. 
	Given its current belief $\mu_t$, the platform generates
	candidates by sampling from a distribution $p(\cdot; \mu_t)$.
	 $p$ the \underline{\smash{algorithm}} that the hiring platform uses. 
	For instance, $p$ might select the candidates it believes the company will 
	be most likely to hire (according to $\mu_t$), 
	or $p$ could be a bandit-like algorithm that ensures exploration.
\end{example}

\begin{example}[Uber example]\label{ex:uber_1}
	Suppose
	drivers use a ride-sharing app---say, Uber---to find riders. 
	At each time step,
	Uber (i.e., the platform) proposes a ride $Z_t$ and the driver (i.e., the user) decides 
	whether to accept, reject,
	or accept then cancel the ride, 
	as denoted by $B_t$.
	Uber's payoff from the interaction is $V(Z_t, B_t)$,
	which could be a constant fraction of the ride's cost if the driver accepts it.
	Similarly, $U(Z_t, B_t)$ is the driver's payoff,
	which might be zero if they decline the ride and otherwise depend on the ride payment or whether the ride takes the driver closer to
	home,
	among other factors.
	
	Uber's goal is to match drivers and riders. 
	To minimize delays, 
	Uber attempts to learn each driver's habits and preferences $q$. 
	Specifically, it assumes that $q$ (i.e., how the driver
	accepts, declines, or cancels rides) belongs to some class $\hcQ$. 
	After observing the driver's behavior, 
	Uber updates its belief $\mu_t$ about $q$,
	and then uses its belief along with its {algorithm} $p$ 
	to better match riders and drivers.
	For instance, whether the platform maximizes revenue, minimizes cancellations, seeks to keep drivers in their preferred zones, or a combination of these factors is reflected in $p$.
\end{example}

\begin{table}[t]
	\caption{Key concepts and notation from \cref{sec:model,sec:user_strat}.}
	\label{tab:key_notation}
	\begin{tabularx}{\textwidth}{lcX} 
		\toprule
		{\bf Object} & {\bf Symbol} & {\bf Description} \\ \midrule
		Proposition space
		& $\cZ$ & Platform action space, subset of $\mathbb{R}^{d_1}$ (exogeneous) 
		\\ \specialrule{0.1pt}{2pt}{2pt}
		Behavior space & $\cB$ & User action space, finite subset of $\mathbb{R}^{d_2}$ (exogenous)
		\\\specialrule{0.1pt}{2pt}{2pt}
		Payoff functions & $U, V$ & Functions that map $\cZ \times \cB$ to $\mathbb{R}$ (exogeneous) 
		\\\specialrule{0.1pt}{2pt}{2pt}
		User strategy & $q$ & Mapping $q : \cZ \rightarrow \Delta(\cB)$ such that $B_t \sim q(\cdot|Z_t)$ 
		\\ \specialrule{0.1pt}{2pt}{2pt}
		Hypothesis class & $\hcQ$ & Finite set of models $\{\hq_i\!:\! i\! \in\! \Omega\}$ that platform uses to estimate $q$
		\\\specialrule{0.1pt}{2pt}{2pt}
		Platform belief  & $\mu_t$ & Distribution over $\hcQ$   
		\\ \specialrule{0.1pt}{2pt}{2pt}
		Platform algorithm & $p$ & Function that maps a belief $\mu \in \Delta(\hcQ)$ to a distribution in $\Delta(\cZ)$ such that  $Z_t \sim \pmu{\mu_t}{\cdot}$ 
		\\
		\bottomrule
	\end{tabularx}
\end{table}

\vspace*{-1em}
\section{User strategization}
\label{sec:user_strat}
In our model (\cref{sec:model}), 
the user selects their strategy~$q$ after the platform has declared its strategy $\smash{(p, \hcQ)}$. 
Although our model allows {us to analyze} a wide range of user behaviors,
there are two types of users of particular interest.
The first type---a \emph{naive} user---behaves without regard for how their current actions affect future outcomes by playing actions that maximize their immediate payoff.
On the other hand, 
a \emph{strategic} user plans ahead; 
they adapt their {strategy} to the platform's strategy 
$(p, \hcQ)$ in order to elicit high payoffs in the long run.
In this way, 
a strategic user's behavior is {dependent on} $(p, \hcQ)$,
whereas a naive user's behavior remains the same across different choices of $(p, \hcQ)$. 
We formalize these two types of users below and visualize 
{their behavior} in Figure \ref{fig:strategization}.

\subsection{Naive user} \label{sec:BR_user}
At each time step $t$, 
the {naive user} chooses the
action $B_t$ that maximizes
their payoff $U$ under the given proposition $Z_t$, as defined next.
If multiple behaviors $B \in \cB$ maximize the immediate payoff,
we assume the user chooses between them uniformly at random.
\begin{definition}[Naive user]
	\label{def:naive_user}
	The \emph{naive} user adopts the strategy $q^\BR$, defined as
	\[
	q^\BR(B | Z) \propto \ind{B \in \arg \max_{B \in \cB} U(Z, B)} , \qquad \forall\ B \in \cB, Z \in \cZ .
	\] 
\end{definition}
\noindent Importantly, 
a {naive} user's strategy $q^\BR$ is independent of the platform's strategy; 
that is, $q^\BR$ remains the same across all choices of $(p, \hcQ)$.

\subsection{Strategic user} \label{sec:strat_user}
In contrast {to a naive user}, 
a \emph{strategic} user maximizes their {\em long-term expected payoff}. 
A strategic user does this by
first considering how the platform's belief $\mu_t$ evolves as $t \rightarrow \infty$ if the user adopts some strategy $q$. 
The user uses this understanding to predict its payoff under all possible strategies $q$ as $t \rightarrow \infty$, 
then selects a strategy that achieves the highest long-term payoff. 

{We formalize this idea in two stages.}
First, 
we define the notion of a globally stable set, 
which tells us how the platform's beliefs evolve as $t \rightarrow \infty$ for a given $(q, p, \hcQ)$.

\begin{figure}[t]
	\centering
	\includegraphics[width=\textwidth]{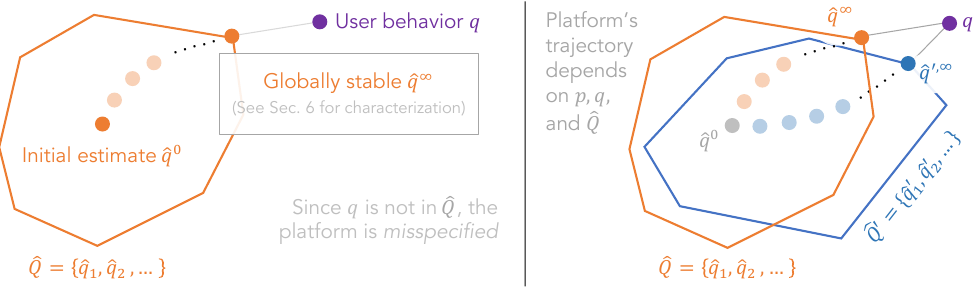}
	\caption{
		{\bf Convergence of platform beliefs about the user as $t \rightarrow \infty$.}
		{ (Left)} Suppose the user adopts strategy $q$,
		and the platform begins with an initial belief $\smash{\mu_0}$. 
		For
		illustrative purposes, we visualize the platform's initial belief using the corresponding estimate 
		$\smash{\hq^0}$,
		and we use the orange polygon to represent $\text{ConvexHull}(\hcQ)$.
		(In this figure, superscripts on $q$ represent time steps, and subscripts index hypotheses/models in $\hcQ$.)
	    As the platform collects data, 
		its estimate evolves,
		eventually converging. 
		The beliefs to which the platform converges is given by the \emph{globally stable set}, as defined in \cref{def:stable};
		in this figure, we visualize the globally stable set $\hcQ_\infty \subset \hcQ$ as a singleton set $\hcQ_\infty = \{\hq^\infty\}$
		such that the platform's limiting belief under $(q, p, \hcQ)$ is the point-mass belief
		$\mu_\infty = \delta_{\hq^\infty}$. 
		{ (Right)} As formalized in \cref{def:stable},
		the 
		belief to which the platform converges 
		depends on the platform's strategy $\smash{(p, \hcQ)}$ and the user's strategy $q$.
		We illustrate this dependence
		by visualizing how changing the platform's hypothesis class 
		(from $\hcQ$ to $\hcQ'$) affects the platform's 
		limiting belief (from $\delta_{\hq^\infty}$ to $\delta_{\hq'^{,\infty}}$).}
		\label{fig:stable}
\end{figure}

\begin{definition}[Globally stable set]\label{def:stable}
	A set $\hcQ_\infty \subset \hcQ$ is $(q, p, \hcQ)$-\emph{globally stable} under
	hypothesis class $\hcQ$, algorithm $p$, and user strategy $q$ if and only if, 
	for any full-support initial belief $\mu_0$,
	$$\bbP (\mu_t ( \hcQ_\infty ) \rightarrow 1 ) = 1 \qquad \text{ as } \qquad t \to \infty,$$ 
	where the probability above is taken with respect to the
	dynamics given in \cref{sec:model}.
\end{definition}
A set $\hcQ_\infty \subset \hcQ$ is
\emph{globally stable} if and only if it contains the support of the platform's
limiting belief (i.e., of $\mu_t$ as $t \rightarrow \infty$) under strategies $(q, p, \hcQ)$.
In Section \ref{sec:main_results}, we discuss one way to ascertain the globally stable set, 
i.e., the platform's limiting belief.

Next, we define the platform and user's expected payoffs.
\begin{definition}[Expected payoffs]
	Consider a distribution $r \in \Delta(\cZ)$ over propositions $\cZ$
	and a user strategy $q \in \cQ$. 
	Then, the platform's and user's \emph{expected payoffs} under $(r, q)$ are
	\begin{align}
		\brV(r, q) &\coloneqq \bbE \left[ 
		V(Z, B) \right], \nonumber
		\\
		\brU(r, q) &\coloneqq  \bbE
		\left[ U(Z, B) - \lambda \cdot d_{TV}(q(\cdot|Z) , q^\BR(\cdot|Z) ) \right] , \label{eq:U_bar}
	\end{align}
	where the expectations are taken with respect to $Z \sim r$ and $B | Z \sim q(\cdot | Z)$, 
	$d_\cQ$ is {some} distance metric
	over $\Delta(\cB)$, 
	and $\lambda \geq 0$. 
	The penalty term $\lambda \cdot d_{TV}(q(\cdot|Z), q^\BR(\cdot|Z))$ 
	in \eqref{eq:U_bar} captures
	the effort that the user expends to deviate from their naive (best-response) behavior.
\end{definition}
Equipped with these definitions, we can now define the {strategic user} as
a user who maximizes their expected payoff under the platform's worst-case,
limiting behavior.
Since different choices of $(q, p, \hcQ)$ {can} induce different globally 
stable sets (\cref{def:stable}),
we define the strategic user with respect to a {\em function} 
$S(q, p, \hcQ)$ that maps $(q, p, \hcQ)$ to {a} corresponding globally stable
set.

\begin{definition}[Strategic user]
	\label{def:S-strat-user}
		{Let}
		$S(q, p, \hcQ)$ 
		{be a function} 
		that maps a user strategy $q$ and platform strategy $(p, \hcQ)$ to a $(q, p, \hcQ)$-globally stable set $\hcQ_\infty$, as defined in \cref{def:stable}.
	{Then, we define the $S$-strategic user as a user who adopts the strategy $q_S^\star(p, \hcQ)$, where}
	\begin{align}
		\label{eq:S-strat}
		q_S^\star(p, \hcQ) & \in \arg \max_{q \in \cQ} 
		\min_{\mu \in \Delta(S(q, p, \hcQ ))}  \brU ( \pmu{{\mu}}{\cdot} , q ).
	\end{align}
\end{definition}

\paragraph{Understanding the definition.}
To tease apart \cref{def:S-strat-user}, 
consider each component of \eqref{eq:S-strat}.
First, recall {from Definition \ref{def:stable}} that a set of user models $\hcQ$ is {\emph{globally stable}} if it contains all the user models to which the platform assigns positive probability as $t \rightarrow \infty$. 
That is, if the user and platform adopt strategies $(q, p, \hcQ)$
and {$S$ is as defined in Definition \ref{def:S-strat-user}}, 
the platform's beliefs as $t \rightarrow \infty$ are contained in $\Delta(S(q, p, \hcQ))$.
Second, 
note that $\brU(\pmu{{\mu}}{\cdot}, \cdot)$ is the user's expected payoff if the platform uses algorithm $p$ to generate propositions under belief $\mu$. 
Putting these two observations together, 
\[
	\min_{\mu \in \Delta(S(q, p, \hcQ ))}  \brU ( \pmu{{\mu}}{\cdot}, q ) ,
\] 
is the $S$-strategic user's envisioned worst-case expected payoff as $t \rightarrow \infty$. 
The user then chooses the strategy $q \in \cQ$ that maximizes this worst-case, limiting payoff.
Therefore, 
{a {$S$-}strategic user maximizes their worst-case limiting payoff under the chosen strategies $(q, p, \hcQ)$ and mapping $S$.}

Importantly, a strategic user pays attention to the platform's strategy $(p, \hcQ)$ whereas a naive user's $q^\BR$ is the same regardless of the platform's chosen strategy.
We illustrate the differences between naive and strategic users in
Figure \ref{fig:strategization}.

\begin{figure}[t]
	\centering
	\includegraphics[width=\textwidth]{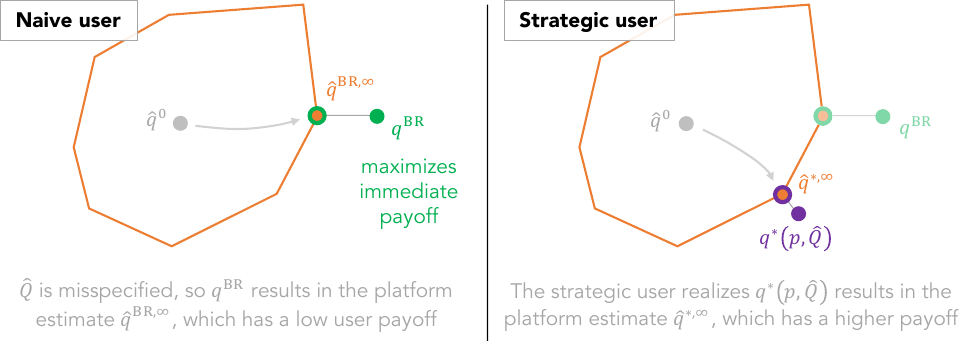}
	\caption{
		{\bf Illustration of a naive user (Section \ref{sec:BR_user}) and a strategic user (Section \ref{sec:strat_user}).}
		{ (Left)}
		The (convex hull of the) platform's hypothesis class
		$\hcQ$ 
		is given by the orange polygon.
		The naive user's strategy $q^\BR$ is given by the solid green dot.
		As in Figure \ref{fig:stable}, 
		the platform's estimate of $q^\BR$ evolves as $t \rightarrow \infty$;
		we visualize the limiting estimate as $\hq_i = \hq^{\BR, \infty}$.
		{ (Right)} 
		The strategic user considers their payoff under the platform's limiting estimate,
		i.e., $\smash{\protect \brU(p^{\delta_{\hq^{\BR, \infty}}}, q^\BR)}$ and finds that they can instead adopt the strategy $q^*(p, \hcQ)$ that leads the platform to a belief (and in turn, a proposition distribution) that is more favorable
		for the user, 
		i.e., $\smash{\protect \brU(p^{\delta_{\hq^{*, \infty}}}, q^*(p, \hcQ))} > \smash{\protect \brU(p^{\delta_{\hq^{\BR, \infty}}}, q^\BR)}$. 
	}
	\label{fig:strategization}
\end{figure}

\paragraph{Remarks on globally stable sets.}
A few remarks are in order. 
First, there is no unique globally stable set, in general. 
For instance, the entire {hypothesis class} $\hcQ$ is trivially {a} globally stable {set}  because $\mu_t(\hcQ) = 1$ for all $t$. 
More fine-grained stable sets  provide stronger analyses. 
We will show in Section \ref{sec:main_results} that there is a principled way to obtain 
{rather fine-grained} globally stable sets.

Second, 
in some cases, 
we will show the existence a globally stable {\em singleton} set, 
i.e., $S^\infty(q, p, \hcQ) = \{\hq^\infty\}$. 
Figure \ref{fig:stable} illustrates on such example.
In this case, 
the platform converges (almost surely) to the point mass 
belief
$\mu_\infty = \delta_{\hq^\infty}$ that depends on $(q, p, \hcQ)$,
and thus that the platform
generates propositions 
from $\pmu{{\mu_\infty}}{\cdot}$. 
As we discuss later on, 
we cannot always guarantee that the platform converges to a single unique belief $\mu_\infty$.
When a unique limiting belief does not exist, 
the globally stable set is our next best tool, 
as it characterizes the set of \emph{possible} limiting beliefs.

Finally, even though strategization is defined with respect to a function $S$,
when $S$ is clear from context we omit it and say ``strategic user'' instead of ``$S$-strategic user.''
Similarly, we omit $S$ or $\hcQ$ from our notation for the strategic user's strategy $q^\star_S(p, \hcQ)$ 
(see \eqref{eq:S-strat}) when clear from context.

\section{Stylized Example}
\label{sec:example}
What are the implications of user strategization? Is it good or bad for platforms?
In this section, 
we consider a styled setting that allows us to answer these questions.
Our goal is to 
illustrate and motivate our main results,
which we establish in their full generality in Section \ref{sec:main_results}. 
	
	At a high level, we consider a simple recommender system that partitions its users into ``types''
	(e.g., comedy lovers and horror lovers) 
	and recommends content based on these types. 
	The platform's payoff is high when the user engages with the recommendations, 
	and the user's payoff depends on both their engagement and their personal taste. 
	Within this setting, 
	we show that:
	\begin{enumerate}[itemsep=0pt]
		\item
		The user is incentivized to be strategic, i.e., there is a user
		strategy that guarantees higher payoff than the
		naive (best-response) strategy $q^\BR$ defined in \cref{def:naive_user}
		(\cref{prop:illustrative_strat_happens}).
		
		\item
		Whenever the recommendation strategy $(p, \hcQ)$ is fixed, 
		the platform's payoff is \emph{never lower} when the user 
		behaves strategically than it is when the user is naive
		(\cref{prop:illustrative_improve_utility}).
		
		\item At the same time,
		user strategization can indirectly hurt the platform. 
		Specifically, suppose that the platform wishes to update its recommendation algorithm from $p$
		to a counterfactual algorithm $p_\text{CF}$. 
		We show that the data that a platform obtains under $p$ cannot be used to reliably make inferences under $\pcf$
		(\cref{prop:illustrative_cfx_bad}).
		
		\item In the same vein, user strategization makes it harder to predict the effect of design choices. 
		Specifically, we show that, when the user is strategic, 
		expanding the hypothesis class $\hcQ$ can unexpectedly hurt the platform's payoff.
		That is, if the platform uses a richer class of user models to estimate the user's preferences, 
		the platform's payoff can actually \emph{decrease} when the user is strategic. 
		(\cref{prop:illustrative_expand_omega})
	\end{enumerate}

\begin{figure}[t]
	\centering
	\includegraphics[width=.95\textwidth]{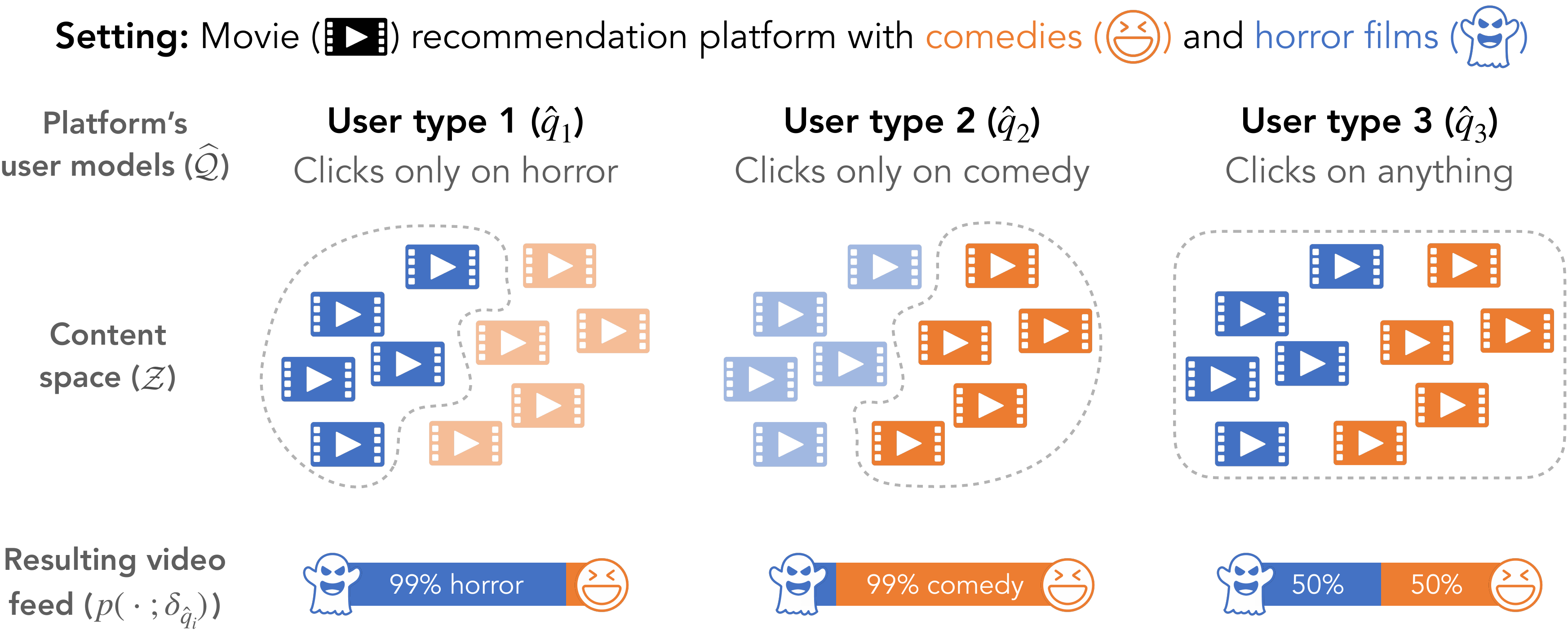}
	\caption{{\bf The recommender system that we consider 
	in our stylized example (Section \ref{sec:example}).}
	The platform's hypothesis class consists of three 
	user models. Under one model, the user watches exclusively
	horror movies; under the other, exclusively comedies; and under the last 
	model, the user is equally interested in comedy and horror.
	The platform represents the user as a convex combination 
	of these models, which dictates the recommendations 
	that the platform gives.
	For example, as shown at the bottom of the figure, 
	if the platform believes the user is of type 1, 
	then the platform shows the user horror movies with 99\% probability. 
	}
	\label{fig:illustrative_setup}
\end{figure}

\subsection{A simple recommender system}
\label{sec:illustrative_setting}
 
Consider a platform that recommends from a finite set of items
$\cZ$ and allows the user to click or ignore the recommended item (i.e., let $\cB = \{0, 1\}$).
Let the platform's and user's payoffs be 
\begin{equation}
	\begin{aligned}
		V(Z, B) = B ,  \qquad
		U(Z, B) = B \cdot a(Z),
	\end{aligned}	\label{eq:user_payoff}
\end{equation}
for all $Z \in \cZ$ and $B \in \cB$, 
where $a : \cZ \rightarrow \{-1, 1\}$ is a fixed function that encodes 
the user's affinity for item $Z$.
As such, the platform gains utility when the user clicks on any recommendation,
and the user gets (possibly negative) utility $a(Z)$ from clicking on an item $Z$, 
and $0$ from not clicking.
We summarize our setup in \cref{fig:illustrative_setup}.

\paragraph{User types and the platform's hypothesis class.}
Suppose that there are two disjoint types of content on the platform,
$\cZ_1$ and $\cZ_2$ (e.g., horror movies and comedies) 
with $\cZ_2 = \cZ \setminus \cZ_1$.
Let there be three types of users: those who prefer $\cZ_1$ 
(i.e., $a(Z) \leq \ind{Z \in \cZ_1}$),
those who prefer $\cZ_2$
(i.e., $a(Z) \leq \ind{Z \in \cZ_2}$),
and those who fall into neither of the two former categories.
Note that $a(Z) \leq \ind{Z \in \cZ_1}$ means that the user definitely does not enjoy anything outside of $\cZ_1$ and may sometimes enjoy content in $\cZ_1$. 

The platform is aware of the three user types, 
but it does not {know} the correct partitioning 
$(\cZ_1, \cZ_2)$ 
and instead believes that the two kinds of content are $(\cZ_A, \cZ_B)$.
This setting is common and can be generalized to capture instances in which the platform does not account for all possible users 
(e.g., there are ``minority'' users that do not follow mainstream trends).

The platform thus
estimates user behavior using the hypothesis class $\hcQ = \{\hq_1, \hq_2, \hq_3\}$, 
where
\begin{equation}
	\label{eq:parameters}
	\hq_i(B=1|Z) = \begin{cases}
		(1 - \gamma)\ind{Z \in \cZ_A} &\text{if } i = 1 , \\
		(1 - \gamma)\ind{Z \in \cZ_B} &\text{if } i = 2 , \\
		(1 - \gamma)\ind{Z \in \cZ} &\text{if } i = 3
	\end{cases}
	\qquad \forall\ Z \in \cZ ,
\end{equation}
and $\gamma > 0$ is a constant that captures the fact that users 
will not click on {\em all} items of a given type.

{Intuitively, 
the platform believes there are three possible users:
users tend to like either only content in $\cZ_A$, 
or only content in $\cZ_B$, or all content. 
Formally, this means that under the user model $\hq_1$ the user only clicks on items in $\cZ_A$ and does so with probability $1 - \gamma$. 
Under $\hq_2$, the user behaves analogously toward $\cZ_B$. 
Under $\hq_3$, the user clicks on any item with probability $1 - \gamma$.}
\paragraph{Recommendation algorithm.} 
Finally, suppose that the platform uses a simple
algorithm $p$ that with small probability $\eps$ (which we specify later) 
recommends a random item $Z \in \cZ$, and
otherwise recommends $Z_i$ with probability proportional to its likelihood
of inducing a ``click'' from the user (under the platform's current belief about the user). 
That is, for all $Z \in \cZ$, 
\begin{align}
	\pmu{\mu}{Z} \coloneqq
	\underbrace{ \varepsilon \cdot 
		\frac{1}{|\cZ|}}_{\text{Uniform w.p. $\varepsilon$}} + 
	\underbrace{(1 - \varepsilon) \cdot 
	\frac{\sum_{\hq_i \in \hcQ} \mu(\hq_i) \cdot \hq_i(B=1|Z)}
		 {\sum_{Z \in \cZ} \sum_{\hq_i \in \hcQ}\mu(\hq_i) \cdot \hq_i(B=1|Z)}}_{\text{
			Proportional to click probability $\hq(B=1|Z)$ w.p. $1-\varepsilon$
		 } },
	\label{eq:rec_alg_1_ex}
\end{align}
where recall from \cref{sec:model} that $\mu$ is the platform's belief, 
and so $\mu(\hq)$ is the probability that the platform assigns to the 
user's strategy being $\hq$.
We visualize this setup in \cref{fig:illustrative_setup}.

\subsection{User behavior}
Now, consider a user of the first type,
i.e., 
a user for which 
$a(Z) \leq \ind{Z \in \cZ_1}.$
Define the set $\cZ^+ = \{Z: a(Z) = 1 \} \subset \cZ_1$ 
as the items that the user enjoys.  
Recall that a user presented with a recommendation $Z_t$ responds with behavior $B_t$ sampled from their behavior 
strategy $q(\cdot|Z_t)$, 
where this behavior strategy depends on whether the user is 
{\em naive} or {\em strategic}, 
as follows. 

\paragraph{Naive users.} 
In this setting, a naive user 
will click on items for which they have 
positive affinity (i.e., for which $a(Z) = 1$), 
and ignore items for which they have negative 
affinity, i.e.,
\begin{align}
	q^\BR(B=1|Z) = \bm{1}\{a(Z) \geq 0\}\ \forall\ Z \in \cZ.
\end{align}

\begin{figure}[t]
	\includegraphics[width=\textwidth]{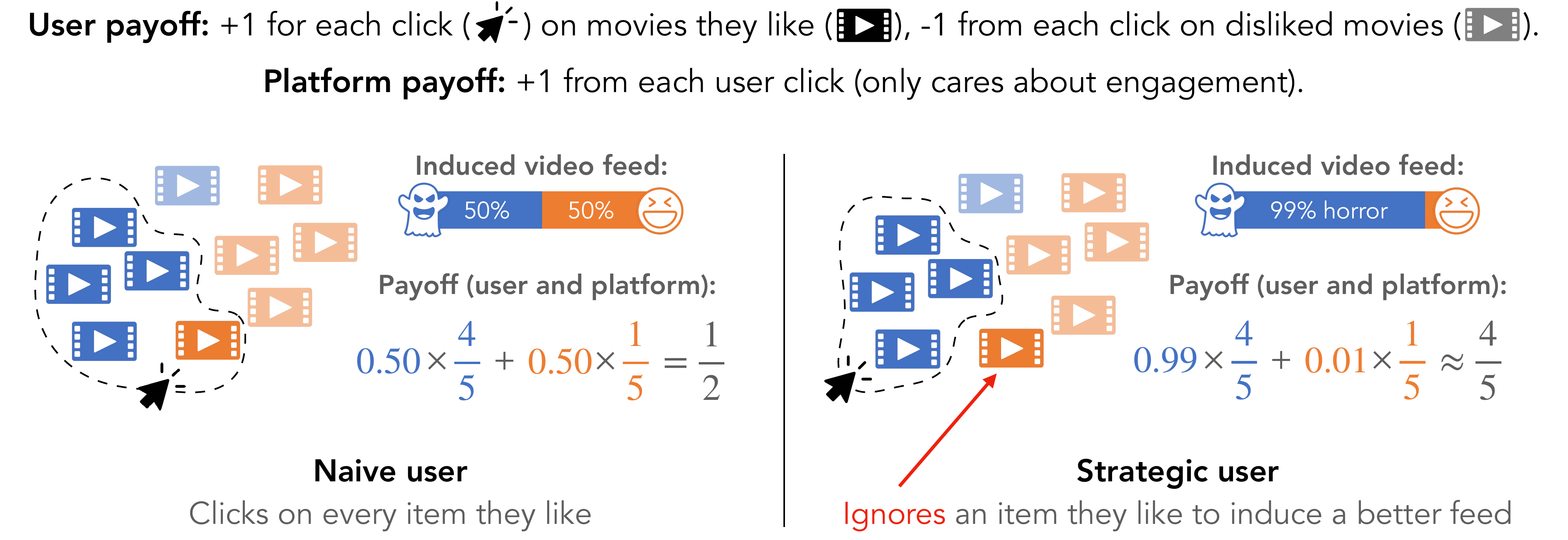}
	\caption{{\bf Naive and strategic user strategies in our stylized example (Section \ref{sec:example}).} 
	We consider a user whose affinity $a(Z)$ is encoded by items' 
	opacity in the Figure above.
	A naive strategy for this user (left) would click on item $Z$ 
	if and only if $a(Z) = 1$. 
	This strategy would result in the platform modeling 
	the user as the ``clicks on anything'' user $\hq_3$ 
	(see \cref{fig:illustrative_setup}),
	and thus serve a feed that is 50\% comedy and 50\% horror.
	If the user is strategic (right), they recognize that the naive 
	strategy is suboptimal, and they avoid clicking on ``outlier''
	comedy videos that they enjoy. The platform thus estimates 
	the user as a ``clicks only on horror'' user $\hq_1$, 
	and serves a feed that better suits the user.
	Notably, {\em both user and platform payoffs are higher} 
	when the user is strategic.}
	\label{eq:chungus}
\end{figure}

\paragraph{Strategic users.}
A strategic user, on the other hand, chooses a strategy that elicits 
the highest long-term payoffs.
For example, they might not click on an item $Z$ 
that they like (i.e., for which $a(Z) = 1$)
in order to influence the distribution of recommended 
items in the future.
To characterize the strategic user, we first need to 
understand the limiting beliefs of the platform.
\begin{restatable}{proposition}{illustrativestable}
	\label{prop:limiting_illustrative}
	Let $\hcQ$ be the hypothesis class defined by
	\eqref{eq:parameters}. Consider a user strategy $q$
	and a platform algorithm $p$ such that $\pmu{\mu}{\cdot}$ has full support for all $\mu \in \Delta(\hcQ)$. 
	Let $\text{supp}(q) = \{Z \in \cZ: q(B=1|Z) > 0\}$. 
	If $|\text{supp}(q)| > 0$,
	then the following function $S$ maps $(q, p, \hcQ)$ to 
	a globally stable set:
	\begin{align*}
		S(q, p, \hcQ) \coloneqq \{\hq_{i^\star}\} , \qquad \text{where} \qquad i^\star = \begin{cases}
			1 &\text{if } |\text{supp}(q) \cap \cZ_B| = 0, \\
			2 &\text{if } |\text{supp}(q) \cap \cZ_A| = 0, \\
			3 &\text{otherwise.} 
		\end{cases}
	\end{align*}
	In other words, the platform's limiting belief is $\mu_\infty = \delta_{\hq_{i^\star}}$.
\end{restatable}
\begin{proof}
	See \cref{app:limiting_illustrative_pf}.
\end{proof}

\noindent 
With this established,
we show that {unless the content that the user likes belongs entirely to $\cZ_A$ or to $\cZ_B$ (i.e., $\cZ^+ \subset \cZ_A$ or $\cZ^+ \subset \cZ_B$),}
then strategization {\em strictly} improves the user's utility.
\begin{restatable}{proposition}{illustrativestrathappens}
	\label{prop:illustrative_strat_happens}
	Consider the setting described in \cref{sec:illustrative_setting},
	and suppose that the platform's partition $(\cZ_A, \cZ_B)$ is not 
	``orthogonal'' to the user's preferences, i.e.,
	\[
		\frac{|\cZ^+ \cap \cZ_A|}{|\cZ_A|} 
		\neq 
		\frac{|\cZ^+ \cap \cZ_B|}{|\cZ_B|}.
	\]
	Then, for sufficiently small $\eps$ in \eqref{eq:rec_alg_1_ex}, 
	a strategic user's $q^\star = q^\BR$ if and only if $\cZ^+ \subset \cZ_A$ or $\cZ^+ \subset \cZ_B$.
\end{restatable}
\begin{proof}
	See \cref{app:illustrative_strat_happens_proof}. 
	In short, when the user is naive and the platform is misspecified, 
	the platform  
	learns $\mu = \delta_{\hq_3}$ by \cref{prop:limiting_illustrative}
	and thus recommends a uniform distribution over all items. If either $\cZ_A$ or $\cZ_B$ 
	is closer to $\cZ_1$, 
	the strategic user can improve their utility by restricting their clicks to only the set
	$\cZ^+ \cap \cZ_A$ (which would lead to the platform recommending from 
	$\cZ_A$) or $\cZ^+ \cap \cZ_B$ (which would lead to the platform recommending from $\cZ_B$).
\end{proof}
\noindent 
\subsection{User strategization improves platform payoffs}
We now examine the impact of strategic behavior on the platform.
We begin by showing that, in our stylized example, user strategization improves 
(or at least does not hurt) the platform's payoffs. 
\begin{restatable}{proposition}{illustrativeimproveutility}
	\label{prop:illustrative_improve_utility}
	Consider the setting described in \cref{sec:illustrative_setting}.
	The platform's payoff is as least as high when the user is strategic as when
	the user is naive.
\end{restatable}
\begin{proof}
	See \cref{app:illustrative_improve_utility_proof}.
Informally, the user only strategizes to get ``better''
content in the long term,
and the platform benefits from this behavior because it receives a positive payoff every time the user clicks. 
\end{proof}
\noindent Note that \cref{prop:illustrative_improve_utility} holds 
regardless of the platform's choice of strategy $(p, \hcQ)$, 
and depends only on the structure of the payoff functions $U$ and $V$.

\subsection{User strategization results in unexpected phenomena}
Suppose that, having deployed algorithm $p$ 
and converged to the limiting belief 
$\mu_\infty$
given by \cref{prop:limiting_illustrative},
the platform considers changing its 
algorithm to downweight toxic content.
Specifically, the platform considers replacing $p$ with 
a counterfactual algorithm $\pcf$:
\begin{align}
	\pcf(Z;\mu)
	\propto \textsc{toxicity}(Z) \cdot \pmu{\mu}{Z}, \label{eq:alt_alg_toxic}
\end{align}
where for some $\alpha \in (0, 1)$, the function $\textsc{toxicity} (Z) \in
\{\alpha, 1\}$ discounts content $Z$ that is toxic.
{We show that strategic behavior (on the part of the user) 
can \emph{corrupt the platform's data} so that
data gathered under $p$ cannot be used to make reliable inferences under
$\pcf$.}

\paragraph{Unreliable counterfactual inferences.}
A natural way for the platform to gauge whether using the algorithm 
$\pcf$ is a good idea is to {\em predict} its (counterfactual) 
payoff under $(\pcf, \hcQ)$. 
Formally, the platform tries to estimate $\smash{\brV^\star(\pcf, \hcQ )}$, 
where for any $p$ we define $\brV^\star(p , \hcQ)$ as
\begin{equation}
	\label{eq:true_payoff_illustrative}
	\brV^\star(p , \hcQ) \coloneqq  \max_{\mu \in \Delta(S(q^\star , p, \hcQ))}\brV(\pmu{\mu}{\cdot}, q^\star(p, \hcQ) ) .
\end{equation}
We recall that $\smash{q^\star(p, \hcQ)}$ characterizes how the strategic user behaves 
in response to the platform strategy $\smash{(p, \hcQ)}$.
Of course, the platform does not have access to $\smash{\brV^\star(\pcf, \hcQ)}$, 
as it does not know 
what the user's strategy {\em would be} in response to $\smash{(\pcf, \hcQ)}$, 
and so it must instead predict its payoff using its current user model
$\hq_{i^\star}$ (collected under $p$). In other words, it computes
\begin{equation}
	\label{eq:est_utility_illustrative}
	\hV(\pcf, \hcQ) \coloneqq \brV(\pcf(\cdot;\delta_{\hq_{i^\star}}), \hq_{i^\star}),
\end{equation}
where we recall that $\hq_{i^\star}$ is the platform's limiting belief from playing 
the algorithm $p$.

We show below that for some choices of the toxicity function, 
the above approach \eqref{eq:est_utility_illustrative} produces a drastically wrong estimate of the platform's 
long-run utility under $\pcf$. 
In particular, this mis-estimation might cause the platform to think that
switching to $\pcf$ would lower the platform's payoffs when it would actually do the opposite. 
\begin{restatable}{proposition}{illustrativecfxbad}
	\label{prop:illustrative_cfx_bad}
	Consider the setting described in \cref{sec:illustrative_setting},
	and the counterfactual algorithm $\pcf$ given by~\eqref{eq:alt_alg_toxic}.
	For any content partitioning $(\cZ_A, \cZ_B)$ where $|\cZ_A|, |\cZ_B| \geq 4$,
	there exists an affinity function $a(Z)$
	(see \eqref{eq:user_payoff}), 
	constants $\gamma > 0$ (see \eqref{eq:parameters}),  $\eps > 0$ (see \eqref{eq:rec_alg_1_ex}),›
	and a function 
	$\textsc{toxicity} : \cZ \to \{\alpha, 1\}$ such that applying \eqref{eq:est_utility_illustrative} simultaneously yields the following outcomes:
	\begin{enumerate}
		\item[(a)] the platform correctly predicts its own utility under $p$, i.e., $\hV(p, \hcQ) = \brV^*(p, \hcQ)$;
		\item[(b)] the platform thinks its payoff will decrease if it switches to algorithm $\pcf$, i.e., \(\hV(\pcf, \hcQ) < \brV^*(p, \hcQ)\);
		\item[(c)] in reality, the platform's payoff will increase if it switches to $\pcf$, i.e., \( \brV^\star(\pcf, \hcQ) > \brV^\star(p, \hcQ). \)
	\end{enumerate}
\end{restatable}
\begin{proof}
	See \cref{app:illustrative_cfx_bad_proof}.
    Intuitively,  when the user strategizes and only 
	engages with items from $\cZ^+ \cap \cZ_A$, the platform has
	no information about how much the user likes items from $\cZ_B$. 
	If, 
	under $p_\text{CF}$, the user switches to engaging only with 
	items from $\cZ^+ \cap \cZ_B$ 
	(say, due to elements in $\cZ^+ \cap \cZ_A$ being marked as toxic), 
	the platform's predicted payoffs will be far from the true ones. 
\end{proof}

\noindent We generalize (and in fact, strengthen) this result in
\cref{sec:main_results} (\cref{prop:cf_p}).

\paragraph{Expanding $\hcQ$ can hurt the platform when users are strategic.}
To further illustrate the counterintuitive phenomena induced by user strategization,
we also demonstrate that expanding the hypothesis class of user models can actually {\em
lower} the platform's payoff.
This finding is unexpected, 
as considering a richer class of models typically does not hurt a learning 
algorithm.
\begin{restatable}{proposition}{illustrativeexpandomega}
	\label{prop:illustrative_expand_omega}
	Consider the setting described in \cref{sec:illustrative_setting}.
	For any partitioning $(\cZ_A, \cZ_B)$ of $\cZ$ 
	there exists an affinity function $a(Z)$ (see \eqref{eq:user_payoff}),
	constants
	$\gamma > 0$ (see \eqref{eq:parameters}),
	  $\eps > 0$ (see \eqref{eq:rec_alg_1_ex}),
	and a user model $\hq_4$ such that when $\hq_4$ is added to hypothesis 
	class $\hcQ$, the platform's payoff under  user strategization~decreases.
\end{restatable}
\begin{proof}
	See \cref{app:illustrative_expand_omega_proof}.  Intuitively, the platform can inadvertently 
eliminate the user's {\em means of strategization},
by adding a candidate user model $\hq_4$ that is more similar to $q^\star$ 
than $\hq_{i^\star}$ is to $q^\star$, but whose corresponding
distribution over content $p(\cdot; \delta_{\hq_4})$ is unfavorable for the user.
This incentivizes the user to switch to a strategy that lowers the platform's
payoff.
\end{proof}

\noindent We prove a general version of this result in \cref{sec:main_results}
(\cref{prop:cf_omega}).

\section{User strategization and its discontents}
\label{sec:main_results}
\noindent 
In this section, we present our main results.
These results generalize and strengthen our findings from Section \ref{sec:example}
and can be summarized as follows. 
\begin{enumerate}
	\item In \cref{sec:instantiate_strat}, 
	we show that one can recover a globally stable set using a result from \citet{frick2020stability}.  
	Recalling the definition of strategization, 
	this result allows us to characterize the platform's beliefs as $t \rightarrow \infty$ and, as a result, how a strategic user behaves. 
	
	\item In \cref{sec:general_strat_helps}, 
	we show that user strategization can \emph{help} the platform. 
	Specifically, when user and platform  payoffs are sufficiently aligned, 
	then user strategization improves both user and platform outcomes under a \emph{fixed} platform strategy $(p, \hcQ)$. 
	
	\item In \cref{sec:strategization_bad}, 
	we find that 
	user strategization can mislead the platform when the platform makes changes to $(p, \hcQ)$. 
	In particular, 
	when a user is strategic, 
	(a) the data that a platform obtains under its current algorithm $p$ cannot reliably predict the platform's payoff under a different algorithm $\pcf$, 
	even if $p$ and $\pcf$ are close; 
	and (b) counter to what the platform might expect, 
	expanding the hypothesis class $\hcQ$ can lower the platform's payoff. 

	\item In \cref{sec:general_BR_good}, 
	we show that, in contrast to when users behave strategically, 
	it is straightforward for the platform to anticipate how changing $(p, \Omega)$ 
	affects its payoff when users behave \emph{naively}.
	This finding begs the question: 
	When are users incentivized to behave naively? 
\end{enumerate}

\subsection{Preliminaries}\label{sec:prelim}
We now introduce 
definitions and assumptions that we use in the remainder of the section. 

\paragraph{Definitions.} 
For any two probability measures
$\Pi_1$ and $\Pi_2$ 
defined on a measurable space $(\Omega, \mathcal{F})$
with probability mass (or density) functions 
$\pi_1$ and $\pi_2$,
the 
Kullback-Leibler divergence and 
total variation
distance are given by
\begin{align*}
	\KL(\Pi_1, \Pi_2) &\coloneqq 
	\mathbb{E}_{x \sim \pi_1}\left[
		\log\left(\frac{\pi_1(x)}{\pi_2(x)}\right)	
	\right] \quad\text{(when $\Pi_1$ is absolutely continuous w.r.t. $\Pi_2$)}, \\
	\text{TV}(\Pi_1, \Pi_2) &\coloneqq \sup_{A \in \mathcal{F}} \left|\Pi_1(A) - \Pi_2(A)\right| .
\end{align*}
Recall that $p: \Delta(\hcQ) \to \Delta(\cZ)$ denotes the platform's 
\emph{algorithm}, i.e., a map between its belief about the user 
and the distribution over propositions it plays. 
We quantify the distance between two algorithms 
using their maximum 
total variation
distance; that is, we let
\[
d_{\cP} (p_1, p_2) \coloneqq \sup_{\mu \in \Delta(\hcQ)} 
\text{TV} (p_1(\cdot;\mu), p_2(\cdot; \mu)).
\]
Recall that $q$ denotes the user's strategy. 
Let $p(\cdot; \mu) \times q$ denote the joint distribution of 
$(Z, B)$ when the platform's propositions are drawn according to $Z \sim p(\cdot; \mu) $ and the user's behavior is given by
$B \sim q(\cdot|Z)$.
Similarly, for any user model $\hq_i \in \hcQ$, 
let  $p(\cdot; \mu) \times \hq_i$ denote  the joint distribution of $(Z, B)$ when $Z \sim p(\cdot; \mu) $ and the user's behavior is given by
$B \sim \hq_i(\cdot|Z)$.

For a fixed proposition distribution $\pmu{\mu}{\cdot}$ and a fixed user strategy $q$,
we say that the user model $\hq_i$ {\em strictly dominates} the user model $\hq_j$ when
$p(\cdot; \mu) \times \hq_i$ better explains $p(\cdot; \mu) \times q$ than
$p(\cdot; \mu) \times \hq_j$ does (in the information-theoretic sense described 
below).

\begin{definition}[Strict KL dominance]\label{def:KL_dom}
	A user model $\hq_i$ {strictly KL-dominates $\hq_j$ at $(p(\cdot; \mu), q)$}, as denoted by $\hq_i \succ^q_{p(\cdot; \mu)} \hq_j$ 
	if and only if 
	$\KL ( p(\cdot; \mu) \times q , p(\cdot; \mu) \times \hq_i )  < \KL ( p(\cdot; \mu) \times q , p(\cdot; \mu) \times \hq_j )$.
\end{definition}

\paragraph{Pervasive assumptions.}
We now lay out the assumptions that we will use in the remainder of this work. 
Throughout the entire section (even when not explicitly mentioned), 
we make the following two assumptions adapted from \citet{frick2020stability}.
\begin{assumption}[\citet{frick2020stability}]
	\label{ass:frick_reg}
	The distributions $\pmu{\mu}{\cdot} \times q$ and $\pmu{\mu}{\cdot} \times \hq_i$
	are continuous Radon-Nikodym derivatives with respect to some $\sigma$-finite measure $\nu$ on 
	$\cZ \times \cB$. When $\cZ$ is discrete (resp., continuous), 
	$\nu$ is a product of the counting (resp., Lebesgue) 
	measure on $\cZ$ and the counting measure on $\cB$. 
	In particular, $\pmu{\mu}{\cdot}$, $q(\cdot|Z)$, and $\hq_i(\cdot|Z)$ 
	are all well-defined probability densities.
\end{assumption}
\begin{assumption}[\citet{frick2020stability}]
	\label{ass:p-cont}
	The platform's recommendation algorithm $p$ and hypothesis class $\hcQ$
	satisfy the following three conditions:
	\begin{enumerate}
		\item (Support). For any user strategy $q$,
		user model $\hq$, and $Z \in \cZ$, $\supp\ q(\cdot |Z) \subset \supp\ \hq(\cdot|Z)$.
		\item (Bounded likelihood ratios). There exists a
		$\nu$-integrable function $h(Z, B)$ such that \[
			\sup_{\mu \in \Delta(\hcQ)} \max_{\hq_1, \hq_2 \in \hcQ} \frac{\hq_1(B|Z)}{\hq_2(B|Z)} \cdot \pmu{\mu}{Z} \cdot q(B|Z) \leq h(B, Z)\text{ for all } B, Z \in \cB \times \cZ.
		\]
		\item (Belief continuity). For each user model $\hq \in \hcQ$, there
		exists a neighborhood $\mathcal{N} \ni \delta_{\hq}$ such that for all
		$Z \in \cZ$ and $\mu \in \mathcal{N}$, the function $p(Z; \mu)$ is
		continuous in $\mu$.
	\end{enumerate}
\end{assumption}
\noindent \cref{ass:frick_reg} simply ensures that the probability distributions 
we are dealing with are well-defined, while \cref{ass:p-cont} establishes (mild)
conditions under which we can characterize the platform's limiting belief (which the rest 
of our analysis relies on).

\paragraph{Regularity assumptions.} In the coming sections, we also 
make certain assumptions about the platform payoff, 
algorithm, 
and hypothesis class being well-behaved.
Rather than being necessary for our negative results, 
these assumptions (\cref{ass:nontrivial_omega,ass:perturbable,ass:well-behaved-p})
actually {strengthen} them. 
In particular, we will show (in \cref{sec:strategization_bad}) that 
\emph{in spite of} these regularity conditions, 
platform payoffs are still poorly behaved and unpredictable when users are strategic. 
Unlike \cref{ass:frick_reg,ass:p-cont}, we explicitly reference the following
assumptions when they are in place.
\begin{assumption}[Payoff landscape]
	\label{ass:perturbable}
	For any distribution $r \in \Delta(\cZ)$ and user behavior $q$, 
	and for any $\eps > 0$, there exists $r' \in \Delta(\cZ)$ such that
	$\text{TV}(r, r') < \eps$ 
	and $\brV(r, q) \neq \brV(r', q)$. 
\end{assumption}
\begin{assumption}[Well-behaved algorithm]
	\label{ass:well-behaved-p}
	We call an algorithm $p: \Delta(\hcQ) \to \Delta(\cZ)$ well-behaved if 
	it maps similar beliefs $\mu$ (in terms of the corresponding user models)
	to similar recommendation distributions $p(\cdot;\mu)$, i.e., if 
	\[
		d_{\cP}(p(\cdot; \mu_1), p(\cdot; \mu_2)) \leq L_{\cP} \cdot \mathbb{E}_{\hq_1 \sim \mu_1,\,\hq_2 \sim \mu_2}\left[
		\max_{Z \in \cZ}\ 
		\text{TV}(\hq_1(\cdot|Z), \hq_2(\cdot|Z))
		\right]
		\qquad \forall\ \mu_1, \mu_2 \in \Delta(\hcQ).
	\]
\end{assumption}
\begin{assumption}[Hypothesis class expansion]
	\label{ass:nontrivial_omega}
	A setup $(\cZ, \cB, V, p, \hcQ)$ satisfies the 
	{\em hypothesis class expansion} assumption if 
	for any fixed user strategy $q$
	any two hypothesis classes $\hcQ_1, \hcQ_2 \subset \hcQ$ such that 
	$\hcQ_1 \subset \hcQ_2$, 
	and a globally stable set function $S$ (as defined in \cref{def:stable}),
	\[
	\min_{\mu \in S (q, p, \hcQ_1)} \brV(p(\cdot; \mu), q) 
	\leq 
	\min_{\mu \in S(q, p, \hcQ_2)} \brV(p(\cdot; \mu), q).
	\]
	In other words, for a fixed user strategy $q$, the platform 
	cannot decrease its payoff by expanding its hypothesis class.
\end{assumption}

\subsection{Platforms converge to beliefs that best approximate user in KL-sense}
\label{sec:instantiate_strat}
We now present an intermediate result that lets us 
characterize user strategization.
Specifically, recall that strategic users plan ahead
using a 
{\em globally stable set} (\Cref{def:stable}), 
a set of user models $\smash{S^\infty \subset \hcQ}$ that almost surely 
contains the support of the limiting belief.
To characterize the globally stable set,
we use a concept introduced by \citet{frick2020stability} known as the
\emph{iterated elimination of dominated states}, wherein, given a subset 
$\smash{\hcQ' \subset \hcQ}$,
one eliminates 
user models $\hq_j$ that are strictly \KL-dominated 
by another user model $\hq_i$ in the set (according to \cref{def:KL_dom}).
\begin{definition}[Iterated elimination of dominated parameters]\label{def:KLdom}
	Consider a platform strategy $(p, \hcQ)$ and user strategy $q$. 
	Define the 
	\emph{elimination operator} $R: 2^{\hcQ} \to 2^{\hcQ}$ as
	\begin{align*}
		R (\hcQ') = 
		\{\hq_j \in \hcQ':\ 
		\not\exists\ \hq_i \in \hcQ' \text{ such that } 
		\hq_i  \succ^q_{p(\cdot; \mu)} \hq_j\ \forall\ \mu \in \Delta(\hcQ')\}.
	\end{align*}
	Let $R^0(\hcQ) = \hcQ$ and recursively define  
	$R^n(\hcQ) = R ( R^{n-1}(\hcQ))$. 
	Finally, define
	$S^\infty ( q, p, \hcQ ) = \cap_{n=1}^\infty R^n(\hcQ)$.
\end{definition}

\citet{frick2020stability} show that the fixed point of the 
elimination operator $S^\infty ( q, p, \hcQ )$ is a globally stable set, i.e.,
the platform converges to the parameters that best approximate $q$ in a KL-sense.

\begin{theorem}[Theorem \cite{frick2020stability}]\label{thm:global_stability}
	For algorithm $p$ and user strategy $q$,
	$S^\infty_{p,q}(\Omega)$ is  $(p, q)$-globally stable.
\end{theorem}
Theorem \ref{thm:global_stability} allows us to characterize strategic users. 
Suppose, for example, that given any user strategy $q$, the set 
$S^\infty ( q, p, \hcQ )$ contains a single element $\hq^*(q)$.
In this case, we can be certain that the 
platform's belief converges to $\delta_{\hq^*(q)}$ as $t \rightarrow \infty$. 
A $S^\infty$-strategic user would then choose the strategy $q$ that 
maximizes their limiting payoff $\smash{\brU(\pmu{\delta_{\hq^*(q)}}{\cdot} , q)}$.

In many cases, $\smash{S^\infty ( q, p, \hcQ)}$ reduces to a single element,
called the \emph{Berk-Nash equilibrium} \cite{esponda2016berk}. 
In general, however, 
$\smash{S^\infty ( q, p, \hcQ )}$ can contain more than one element.
In these cases, 
Theorem~\ref{thm:global_stability} tells us that the platform's proposition distribution as $t \rightarrow \infty$ 
is contained within the set $\{ p(\cdot; \mu) : \mu \in \Delta( S^\infty ( q, p, \hcQ ) )\}$.

\subsection{User strategization can help the platform}
\label{sec:general_strat_helps}
Our first main result shows that user strategization can \emph{improve} the platform's payoff $\brV$. 
This occurs when the payoffs of the user and platform are sufficiently aligned.

To see why this might be the case, consider YouTube. 
YouTube would like to engage and retain users
(by serving users good recommendations)
while ensuring the profitability of their platform. 
Although users on YouTube may not care about the platform's profitability,
they do want good recommendations.
In this way, there is some alignment between the user and platform payoff. 
Users on YouTube often have an idea of how the YouTube algorithm works and, 
in response, adapt to the algorithm in order to elicit better recommendations. 
This strategization can ultimately help the platform by improving their recommendations and therefore increasing user engagement. 

The following result (whose interpretation we discuss below) 
corroborates this relationship, 
showing that user strategization can increase the platform's payoffs.

\begin{restatable}{proposition}{strathelps}
	\label{prop:strat_helps_platform}
	Consider a platform strategy $(p, \hcQ)$ 
	{and suppose} that $U( B, Z)$ {has a unique maximizer in $\cB$} for all $Z$.
	{For a user strategy $q$},  
	{let} $\tV(q)$ {be} the platform's worst-case limiting payoff, 
	\begin{align}\label{eq:prop_strat_helps_1}
		\tV(q) \coloneqq \min_{\mu \in \Delta(S^\infty(q, p, \hcQ))} \brV (p(\cdot; \mu), q) , \qquad \forall q \in \cQ ,
	\end{align}
	where $S^\infty$ is defined in \cref{def:KLdom}.
	Then, user strategization strictly improves the platform's worst-case limiting payoff if
	\begin{align}
		\tV \left( \arg\max_{q\in\cQ} \min_{\mu \in \Delta(S^\infty (q, p, \hcQ))} \brU(p(\cdot; \mu) , q) \right) 
		&>
		\tV \left( \arg\max_{q\in\cQ}\min_{\mu \in \Delta(\hcQ)} \brU(p(\cdot; \mu) , q) \right).
		 \label{eq:strategization_helps_platform2}
	\end{align}
	The same is true if the $\min$ in \eqref{eq:prop_strat_helps_1} is swapped out for a $\max$.
\end{restatable}
\begin{proof}
	See \cref{app:strat_helps_proof}.
\end{proof}
\noindent \cref{prop:strat_helps_platform} 
gives simple conditions under which platform payoff is \emph{strictly higher} when a user strategizes than when a user is naive. 
In particular, 
there are two takeaways from \cref{prop:strat_helps_platform}:
\begin{enumerate}
	\item \emph{Strategization helps platforms when $U$ and $V$ are sufficiently aligned.}
	Observe that the right-hand side of \eqref{eq:prop_strat_helps_1} is identical 
	to the analogous expression on the left-hand side of \eqref{eq:strategization_helps_platform2}, 
	except that $\brV$ is swapped out for $\brU$ in  \eqref{eq:strategization_helps_platform2}. 
	Therefore, 
	when $U$ and $V$ are sufficiently aligned, 
	finding a $q$ that maximizes the worst-case user payoff (i.e., $\min_{\mu \in \Delta(S^\infty (q, p, \hcQ))} \brU(p(\cdot; \mu) , q)$) will also yield a high platform payoff. 

	The following corollary confirms this intuition, showing that user
	strategization can never lower the platform's payoff when $U$ and $V$ are
	perfectly aligned. 
\begin{corollary}
	Under the setup of \cref{prop:strat_helps_platform}, 
	the platform's payoff under a strategic user is at least as high as its payoff under a naive user when $U = V$.
\end{corollary}

	\item \emph{Strategization can be viewed as a form of coordination}. 
This follows from the fact that the left- and right-hand sides of \eqref{eq:strategization_helps_platform2} look similar, 
with the only difference being the set of beliefs over which $\mu$ is minimized. 
The set on the left-hand side is constrained by the user's anticipation 
of how the platform will behave. 
When the user and platform have aligned incentives, 
and the user has a good idea of how the platform will behave,
the resulting ``coordination'' helps both the user and the platform. 	

The following corollary solidifies this intuition, showing that even when user and platform 
have only partially aligned payoffs, they can coordinate (via strategization) to find a region 
of proposition space where their payoffs are aligned.
\begin{corollary}
	Assume the same setting as \Cref{prop:strat_helps_platform}.
	The platform's payoff under a strategic user is at least as high as its
	payoff under a naive user when there exist functions $g: \cB \to \bbR$ and 
	$f: \cZ \to \{-1, 1\}$ such that user and platform payoffs decompose as  
	$U(B, Z) = g(B) \cdot f(Z)$ and $V(B, Z) = g(B)$ respectively.
\end{corollary}

\end{enumerate}

\subsection{User strategization can cause unexpected behavior}
\label{sec:strategization_bad}

Although user strategization can help the platform under its current strategy $(p, \hcQ)$, 
we show that strategization interferes with learning by ``corrupting'' the
data that the platform collects. 
We further show strategization can cause other unexpected behavior---specifically,
expanding the hypothesis class $\hcQ$ that the platform uses can unexpectedly \emph{hurt} the platform under strategization.

\subsubsection{Changing the algorithm $p$ results in unpredictable payoff}

In this section, consider a fixed hypothesis class $\hcQ$.
Suppose that, after deploying algorithm $p$, 
the platform wants to change its algorithm to a counterfactual algorithm 
$\pcf$.
In these cases, the platform may wish to \emph{estimate} their expected 
payoff if they were to switch to $\pcf$.
Formally, for any belief $\mu \in \Delta(\hcQ)$ 
and algorithm $p'$, we define the platform's 
{\em predicted payoff} as 
\begin{equation}
	\label{eq:estimated_payoff}
	\hV(p', \mu) \coloneqq \bbE_{\hq \sim \mu}\left[
		\brV(p'(\cdot; \mu), \hq)
	\right] .
\end{equation}
Now,
consider a user who is $S^\infty$-strategic.
Recall from Section \ref{sec:strat_user} that 
$q^\star\!(\pcf, \hcQ)$ denotes the strategy that the strategic user adopts
if the platform employs strategy $(\pcf, \hcQ)$. 
We can consequently write the platform's worst-case payoff under $\pcf$ and user strategization as
\begin{equation}
	\label{eq:true_payoff}
	\brV^\star (\pcf) \coloneqq \min_{\mu \in \Delta(S^\infty(q, \pcf ))} \brV(\pcf, q^\star(\pcf)).
\end{equation}
The following result shows that
if a platform gathers data under $p$ but is interested in estimating its payoff under an alternative algorithm $\pcf$, 
its estimate may be arbitrarily bad when the user is strategic. This result holds {\em even when the platform is (nearly) correctly specified},
i.e., even for the highly expressive hypothesis class $\hcQ$ defined below:
\begin{definition}[$\eps$-net hypothesis class]
	\label{ass:no_misspec}
	For a finite proposition space $\cZ$ and some $\eps > 0$, 
	let $\Delta_\eps(\cB)$ be an $\eps$-net of $\Delta(\cB)$
	with respect to the $\ell_\infty$ metric.
	An $\eps$-net hypothesis class is the set of all possible mappings from 
	propositions in $\cZ$ to behavior distributions in $\Delta_\eps(\cB)$, i.e.,
	$\hcQ_{\eps} \coloneqq {\Delta_{\eps}(\cB)}^\cZ$ (and thus
	$|\hcQ_\eps| \approx (\frac{1}{\eps})^{|B| \cdot |Z|}$).
\end{definition}
By proving the result below under \cref{ass:no_misspec},
we rule out the case where the platform is
unable to predict counterfactual payoffs simply because it cannot accurately
model the user.
\begin{restatable}{proposition}{badpredictiveness}
	\label{prop:cf_p}
	Consider a given platform strategy $(p, \hcQ)$ and platform payoff function $V$.
	Suppose that $\hcQ$ is an $\eps$-net (\cref{ass:no_misspec}) for some
	sufficiently small $\eps$,
	that $\pmu{\mu}{\cdot}$ has full support for all $\mu$,
	and that \cref{ass:perturbable} holds.
	Define $\zeta$ as the maximum gap in predicted platform payoff, 
	i.e.,
	\begin{align*}
		\zeta(p') = \max_{\hq_1, \hq_2 \in \cQ} \brV(p'(\cdot; \delta_{\hq_1}), \hq_1)
		- \brV(p'(\cdot; \delta_{\hq_2}), \hq_2).
	\end{align*}
	Further assume that there exists $\beta \geq \alpha \geq 0$ such that 
	$\alpha \leq \text{Var}_{p(\cdot; \mu) \times q}[V(B, Z)] \leq \beta$
	for any belief $\mu$ and user strategy $q$.
	Then, for any $\eps_0, \eps_1 > 0$, there exists a $\pcf$ and $U$ such that 
	$d_\cP(p, \pcf) \leq \eps_0$, and
	\begin{align}\label{eq:bad_CF}
		\min_{\mu \in S^\infty( q , p )}
		\left|
		\hV(\pcf, \mu)
		- 
		\brV^\star
		(\pcf)
		\right| 
		\geq \sqrt{\zeta(\pcf)^2 - 4(\beta - \alpha)} - \eps_1.
	\end{align}
\end{restatable}
\begin{proof}
	See \cref{app:nonsmooth_proof}. 
\end{proof}
\noindent If the variance of the platform's payoff does not change much across
beliefs (i.e., if $\alpha \approx \beta$),
the right side of \eqref{eq:bad_CF} is effectively the largest gap that can exist between an estimated payoff and the true payoff. 
As such, \cref{prop:cf_p} states that data collected under $p$ can be unhelpful for making predictions about a counterfactual algorithm.
The intuition behind this result is that, when users strategize, the platform limiting payoff is 
{\em nonsmooth} in $p$, i.e., there exists a $\pcf$ that is 
$\eps$-close to $p$ for which the payoff under strategization is far from that under~$p$.

\subsubsection{Expanding \texorpdfstring{$\hcQ$}{} can hurt the platform} 
Next, we show that when a user is strategic, 
expanding $\hcQ$ can hurt the platform's payoff.
This is somewhat counterintuitive, as $\hcQ$ is the hypothesis class that
the platform uses to infer the user's behavior $q$, 
and expanding one's model family typically does not hurt estimation.  

To make this statement formal, 
suppose that the user is $S^\infty$-strategic and that the 
algorithm $p$ is fixed. 
(Since $p$ is held fixed in this section, we suppress its notation below.)
Let $q^\star(\hcQ)$ denote the strategy that the strategic user adopts if the platform uses hypothesis class $\hcQ$. 
That is, let
\begin{align*}
	q^\star(\hcQ) &\coloneqq \arg \max_{q \in \cQ} \min_{\mu \in S^\infty( q , \hcQ )} \brU (p(\cdot; \mu), q).
\end{align*} 
\begin{restatable}{proposition}{badomega}\label{prop:cf_omega}
	Consider a platform strategy $\!(p, \hcQ)\!$ and
	a platform payoff function $V$ bounded in $[0, 1]$.
	Suppose $(\cZ, \cB, V, p, \hcQ)$ satisfy the expansion 
	assumption (\cref{ass:nontrivial_omega}), and that $V(Z, B)$ has a unique 
	maximizer with respect to $B$ for each $Z \in \cZ$.
	Then, there exist a user payoff function $U$
	and sets $\hcQ_1, \hcQ_2 \subset \hcQ$ such that  
	$\hcQ_1 \subseteq \hcQ_2$, but 
	\begin{align*}
		\min_{\mu \in S^\infty( q^\star(\hcQ_1) , \hcQ_1 )}\brV(p(\cdot; \mu), q^\star(\hcQ_1)) 
		> 
		\max_{\mu \in S^\infty( q^\star(\hcQ_2), \hcQ_2 )}\brV(p(\cdot; \mu), q^\star(\hcQ_2)).
		\label{eq:counterfactual_Omega}
	\end{align*}
\end{restatable}
\begin{proof}
	See \cref{app:cf_omega_pf}. The intuition is that the 
	platform can unintentionally remove the user's {\em means of strategization}.
	That is, a user may
	want to induce a specific belief from the platform without straying
    too far from their best-response strategy. 
    Thus, they purposefully pick a strategy $q^*_1$ that the platform 
	misinterprets in a desirable way.
    When the platform gets ``better'' at capturing their behavior by 
    adding a user model that is close to $q_1^*$, the user is forced to 
    move even further away from their best-response behavior, making things 
    worse for the platform.
\end{proof}

Like \cref{prop:cf_p},  \cref{prop:cf_omega} shows that user strategization can cause unexpected behavior. 
Both imply that strategization makes it difficult to use its data under one strategy $(p, \hcQ)$ make inferences about a different strategy.

\subsection{System behavior is more predictable for best-response users}
\label{sec:general_BR_good}
The previous section establishes that strategization can make it difficult for the platform to predict how the system behaves if either $p$ or $\hcQ$ is changed. 
This is consequential because the platform may wish to use the data that they have collected under $(p, \hcQ)$ to estimate some quantity under a counterfactual $\smash{(p', \hcQ')}$. 
In  this section, 
we show that \emph{these problems do not arise when users are naive}.
In other words, 
when the user plays according to their best-response at each timestep, 
the platform benefits because it is easier for them to make inferences under changes to $(p, \hcQ)$. 

Formally, suppose we have a best-response user,  and fix a platform hypothesis class $\hcQ$.
Recall the platform's payoff estimator \eqref{eq:estimated_payoff} (restated below),
and let $\brV_\BR (\pcf)$ denote the true payoff under $(\pcf , \hcQ)$ 
(i.e., analagous to \eqref{eq:true_payoff})
when the user is naive:
\begin{align*}
	\hV(\pcf, \mu) &= \bbE_{\hq \sim \mu}\left[
		\brV(\pcf(\cdot; \mu), \hq)
	\right] .
	\\
	\brV_\BR ( 
	\pcf
	)
	&= 
	\min_{\mu \in \Delta(S^\infty (  q^\BR , \pcf))}
	\brV 
	\left( 
	\pcf(\cdot; \mu) , 
	q^\BR
	\right) .
\end{align*}
Our first result shows that, 
in contrast to Proposition \ref{prop:cf_p}, 
using data gathered under $p$ to estimate the platform's payoff under $\pcf$ is a good idea for best-response users. 
\begin{restatable}{proposition}{bestresponsepredictable}
	\label{prop:smooth_p_BR}
	Suppose that the platform's strategy $(p, \hcQ)$ is such that
	the hypothesis class $\hcQ$ is an $\eps$-net (\cref{ass:no_misspec}).
	Let $\pcf$ be a counterfactual algorithm that 
	is well-behaved (\cref{ass:well-behaved-p}); then,
	\begin{align}
		\max_{\mu \in \Delta(S^\infty(q^\BR, p))} \left|
		\hV ( \pcf, \mu)
		- 
		\brV_\BR ( 
		\pcf
		)
		\right| 
		\leq \sqrt{\eps} \left((2L_\cP + 1) \sqrt{|\cB|}\right).
	\end{align}
	As a result, by using a sufficiently fine $\eps$-net hypothesis class 
	$\hcQ$, the platform can estimate its payoff under counterfactual 
	algorithms up to arbitrary precision.
\end{restatable}
\begin{proof}
	See \cref{app:best_response_predictable_proof}.
\end{proof}
Next, suppose that the hypothesis class $\hQ$ can change but the algorithm $p$ is fixed. 
In addition, the next result shows that, in direct contrast to Proposition \ref{prop:cf_omega},
payoffs cannot decrease when the hypothesis class $\hcQ$ is expanded and the
user plays naively.
\begin{proposition}\label{prop:smooth_omega_BR}
	Consider a given $(p, \hcQ)$ and $V$.
	Under \cref{ass:nontrivial_omega}, for any $\hcQ' \subset \hcQ$,
	Then 
	$$
	V(\hcQ, q^\BR ) \geq V(\hcQ', q^\BR ) . 
	$$
\end{proposition}
\begin{proof}
	This follows directly from \cref{ass:nontrivial_omega}.
\end{proof}
Propositions \ref{prop:smooth_p_BR} and \ref{prop:smooth_omega_BR} show that incentivizing the user to play their best-response behavior can lead to more reliable data and payoffs for the platform, 
begging the question: When are users incentivized to play naively?

\section{Trustworthy algorithm design}\label{sec:trust_sec}

In the previous sections, 
we found that strategic behavior can hurt the platform. 
Specifically, we showed that user strategization violates a key assumption of most data-driven algorithms that user behavior is exogenous,
which compromises the platform's ability to re-use the data that it collects, 
e.g., to train future algorithms. 
In contrast, 
we found that the platform's data \emph{looks} exogenous if the user behaves naively. 
That is, 
a platform can recover the exogeneity assumption if it encourages users to behave naively.

In this section, we  argue that this analysis suggests that \emph{trustworthy design}  can help platforms. 
We discuss how trustworthiness produces two beneficial outcomes: 
users are not incentivized to strategize---which allows platforms to recover the exogeneity assumption---and users are incentivized to engage with the platform over alternatives. 
We begin in Section \ref{sec:trust_def} with a formal definition of trustworthiness. 
Importantly, we draw a distinction between trustworthiness and strategy-proofness, 
connecting our definition to the concept of individual rationality in mechanism design. 
We then discuss how this definition relates to existing notions of trust in Section \ref{sec:elements_trust}.
In Sections \ref{sec:reasons}-\ref{sec:trust_interventions},
we use this definition to unpack four reasons why users do not trust their platforms and propose two  interventions for improving trustworthiness.

\subsection{Trustworthiness and its effect on the platform}\label{sec:trust_def}

Building on our analysis in Section \ref{sec:main_results},
we present a formal definition of trustworthiness. 

\begin{definition}[Trustworthy]\label{def:trust}
	Let $q^\star$ denote the policy that the $S^\infty_\Omega$-strategic user adopts when the platform employs algorithm $p$. 
	A platform's policy $(p, \hcQ)$ is $\kappa$-\emph{trustworthy} when the user is not incentivized to strategize
	and the naive user's limiting payoff under $(p, \hcQ)$ is at least $\kappa \in \bbR$, 
	i.e., 
	\begin{enumerate}
		\item $\min_{\mu \in \Delta(S^\infty_\Omega(p, q^\star))}  \brU (p^{\mu} , q^\star ) \leq 
		\min_{\mu \in \Delta(S^\infty_\Omega(p, q^\BR))}  \brU (p^{\mu} , q^\BR )$ and
		
		\item $	\min_{\mu \in \Delta(S^\infty_\Omega(p, q^\BR))}  \brU (p^{\mu} , q^\BR ) \geq \kappa$.
	\end{enumerate}
\end{definition}

\paragraph{First requirement of trustworthiness.}
One can think of Definition \ref{def:trust}(1) as follows:
a platform satisfies the first requirement of trustworthiness if the user does not have to strategize because the platform looks out for the user's interests so they do not have to do so themselves. 
When the platform meets this requirement, 
the user might as well play their best-response action. 

As an example, 
consider YouTube.
Suppose that a user likes to alternate between ``junk'' content and ``healthy'' content based on their mood \cite{Kleinberg2022-wy}.
If they were ``truthful,''
then they would always pick the video that fits their current mood.\footnote{Note that this setting can be modeled by allowing $U$ to be stochastic, with a hidden variable that represents the user's underlying stochastic mood.}
Suppose that the platform adopts an algorithm that only accounts for one possible mood and can therefore only model the user as liking ``junk'' or ``healthy'' content, but not both. 
Then, the user is better off behaving strategically; for instance, only using YouTube for one---but not both---of their moods (or maintaining two YouTube accounts, as has been observed anecdotally). 
Intuitively, 
the user does not \emph{trust} that the platform interprets their behavioral data correctly; 
they do not trust that the platform will use their behavioral data to generate good content in the future. 
On the other hand, 
if YouTube is able to correctly parse the user's mood and recommend content to suit both moods, 
the user does not have to behave strategically. 
The user therefore \emph{trusts} the platform to correctly interpret their actions. 
In addition to this example, 
there are many other factors that affect trustworthiness, 
such as a platform's willingness to protect user privacy. 

\paragraph{Second requirement of trustworthiness.}
Definition \ref{def:trust}(2) states that trustworthiness is earned only if the user's expected limiting payoff is at least $\kappa$ when the user plays their best-response action.  
That is, it is not enough that the platform is strategy-proof, as required by Definition \ref{def:trust}(1). 
Trustworthiness additionally requires that behaving naively is sufficiently beneficial for the user. 
Consider, for instance, a platform that adopts a strategy $(p, \hcQ)$ under which the user's limiting payoff is $-1$ no matter what strategy the user adopts. 
Then,  $(p, \hcQ)$ satisfies Definition \ref{def:trust}(1) because the user is not incentivized to strategize since all strategies induce a payoff of $-1$, 
but the system is \emph{not} $\kappa$-trustworthy under Definition \ref{def:trust}(2) for any $\kappa > -1$. 

This requirement echoes \emph{individual rationality}, 
a concept in mechanism design under which agents continue engaging with the platform
(which is known in mechanism design as ``participating in the mechanism'') if it is beneficial for them to do so. 
Using this interpretation,
$\kappa$ determines the expected limiting payoff at which the user does not trust the platform although they may be willing to continue engaging with the platform. 
In this way, 
$\kappa$ captures the user's ability to tolerate the untrustworthy behavior. 
The higher $\kappa$ is, 
the more trust the user places in the platform, 
and the more likely the user is to engage with the platform over alternatives.

\paragraph{Implication of trustworthiness on the platform.}
As shown in Sections \ref{sec:example}-\ref{sec:main_results}, 
trust is central to the platform's ability to collect reliable data. 
Without trust, 
users are incentivized to strategize, 
which is particularly harmful for platforms because it means that the data that they use for multiple purposes---such as training future algorithms or predicting the performance of candidate algorithms---is unreliable. 
More broadly, 
earning user trust is often beneficial to platforms when they rely not only on the continued participation of their users, 
but also on the amount of user engagement.
The more time users spend on the platform, 
the more data the platform collects. 
This facet of trustworthiness is captured by the second requirement in Definition \ref{def:trust}---if a platform is only $\kappa$-trustworthy, but another platform is $\kappa'$-trustworthy for $\kappa' > \kappa$, users may be compelled to spend more time on the other platform (or even switch platforms).
There are, of course, 
occasions %
when a platform is not incentivized to build trust, 
which could occur when the platform benefits so greatly from strategization (enough to overshadow the potential harms of unreliable data)
that platform does not rely greatly on data collection for prediction,
or there is little risk of users leaving the platform.

\subsection{Elements of trustworthiness}\label{sec:elements_trust}

The goal of this section is to place our definition of trustworthiness within the broader discussion of trustworthiness. 
Below, 
we identify several key elements of trustworthiness that appear in the literature on trust \cite{kramer1999trust,hardin2006trust,nissenbaum2001securing}, 
where we examine situations in which ``I'' trust ``you.''
The following are five common elements of trust:
\begin{enumerate}
	\item Trustworthiness does not imply that you and I have perfectly aligned interests;
	only that your behavior takes some of my interests into account.
	
	\item Distrust arises when I expect to incur losses from interacting with you  (unless I behave strategically) and, conversely, 
	trust arises when I expect to gain from interaction.
	
	\item  Trust is inherently relational in that trusting you to look after my interests may depend on who ``you'' and ``I'' are, so that ``you'' are not necessarily universally trustworthy.
	
	\item Trust is generally meaningful only when there are repeated interactions.
	In particular, I only put trust in you if I must rely on you in the future.
	Moreover, I only trust you to take my interests into account if I believe you value a continued relationship with me.
	
	\item Trust involves vulnerability---the possibility of harm---because to trust is to allow another to affect one's interests.
\end{enumerate}
We now connect these elements of trust to our
Definition \ref{def:trust}.
First, 
a user and platform need not have the exact same interests for the results in Section \ref{sec:main_results}, which show that untrustworthiness can be harmful to the platform, to hold. 
Second, 
distrust in our work arises when a user is incentivized to strategize because the platform does not account for the user's interests,
meaning that trust and distrust are linked to gains and losses from interactions. 
Third, 
our formulation of trust is relational;
a user's best-response and strategic behavior are specific to them, 
meaning that a platform's trustworthiness may differ across users. 
Fourth, 
a continuing relationship is pivotal to our formulation of trust.
In fact it matters in two ways:
(i) users only strategize because they can anticipate future interactions, 
and (ii) platforms typically rely on the continued participation of users for success. 
Lastly, 
user payoffs, as modeled in Section \ref{sec:setup}, 
depend on both the user's action as well as the platform's. 
As such, the user's utility is influenced by the platform. 
This point is salient because the risk of being harmed is why trust often matters in data-driven decision-making.

\begin{remark}
Note that there are other factors that may influence trustworthiness, 
such as the credibility of the trusted party, 
their reputation, 
and even whether they are virtuous or reliable. 
These factors are based on the \underline{\smash{perceived qualities}} of the trusted party.
Although perception has a significant impact on trust, 
we leave perception-based factors of trust to future discussions.
By putting perception aside, 
we align more closely with conceptualizations of trust such as Hardin's, 
who states that trust is determined by the \emph{interests} of each party and the desire of one party to take the other party's interests into account in order to foster a \emph{continuing relationship} \cite{hardin2006trust}. 
\end{remark}

\subsection{Reasons why a user may not trust their platform}\label{sec:reasons}

Under Definition \ref{def:trust} of trustworthiness, 
there are several reasons why a user might not trust their platform's policy. 
We discuss three reasons in this section. 
These insights can be used to inform methods for building trustworthiness, 
as discussed in the following section. 

\paragraph{Misspecification.}
We say that a platform is \emph{misspecified} with respect to $q^\BR$ when $q^\BR \not\in \hcQ$. 
Misspecification implies that,
should a user choose to play their best-response actions, 
the platform would not be able to model their behavior perfectly. 
For example, 
misspecification occurs when a platform believes that there are a few canonical ``types'' of users, 
but the user of interest does not fit into any of these types. 
Misspecification often induces users to strategize because the platform's inability to model the user correctly can mean that a user gains more by strategizing
(e.g., by pretending to be one of the canonical types of user).

\paragraph{Hidden, changing state.}
Suppose that a user has different ``modes.''
For instance, 
a user on an online shopping platform sometimes needs clothes for everyday wear and sometimes needs clothes for special occasions. 
Alternatively, 
suppose that there is salient information that is only available to the user but not to the platform. 

In such cases, 
there is a state that changes across time and is hidden from the platform. 
We do not explicitly model this setting in our setup in Section \ref{sec:model};
however, one can encompass such cases by adding a Step 0 to the game in Section \ref{sec:setup} during which a state $x_t$ is drawn at the start of each time step, i.e.,  ``nature'' selects a state $x_t$. 
The user then chooses action $B_t \sim q( \cdot | Z_t, x_t)$ in Step 2, 
and the user receives a payoff $U(Z_t, B_t , x_t)$ in Step 3. 
Because the platform does not have access to salient information $x_t$, 
the user may find that their $x_t$-dependent behavior is misattributed to other factors. 
This misattribution can lead the platform to behave unexpectedly (to the user's detriment), 
therefore prompting the user to strategize.

\paragraph{Algorithm incompatible with user payoffs.}
Recall that $p$ denotes the platform's algorithm and $U$ denotes the user's payoff. 
A user may strategize if $p$ is incompatible with $U$. 
Even if the platform is not misspecified with respect to the best-response strategy $q^\BR$ and the platform has access to all salient information (such that the two reasons for strategization given above are absent), 
a user may wish to strategize if the platform's method for generating propositions is detrimental to a user following $q^\BR$. 
This would, for instance, 
be the reason why users often strategize when an algorithm is known to have feedback loops or why users strategize when their recommendation algorithm shows too much content of type $X$ after a user clicks on a piece of content of type $X$.

\subsection{Interventions for improving trustworthiness}\label{sec:trust_interventions}

Data-driven platforms can determine whether users strategize by testing the exogeneity assumption. 
For example, 
 \citet{empirical} determined through a lab experiment that (i) users strategize on online platforms and (ii) a user's strategic behavior reflects how they believe their current actions affect their long-term outcomes when their platform is data-driven. 

When exogeneity does not hold (i.e., users strategize), 
platforms may wish to intervene. 
In this section, 
we first describe why interventions that do not build trustworthiness (as formalized by Definition \ref{def:trust}) but attempt to overcome the challenges untrustworthiness often fall short. 
We then discuss how two interventions for improving trustworthiness can complement these efforts.

\subsubsection{Naive interventions}

Recall that the main challenge of untrustworthiness is that user strategization distorts the platform's data, 
which compromises its ability to estimate user behavior under alternate strategies $(p, \hcQ)$.
The platform might wish to overcome the challenges of untrustworthiness by:
(i) designing a strategy-proof mechanism, 
(ii) modeling the user's payoff function $U$ in order to predict their behavior under any counterfactual $(p, \hcQ)$,
(iii) expanding the hypothesis class $\hcQ$, 
(iv) guessing the user's hidden state $x_t$, as described in Section \ref{sec:reasons}, 
or (v) improving one's algorithm $p$. 
Below, we discuss how these approaches 
can fall short of methods that directly boost trustworthiness.

First, 
designing strategy-proof mechanisms does not necessarily elicit higher payoffs for the user or the platform. 
For instance, 
an algorithm $p$ that recommends YouTube videos that all users universally dislike is strategy-proof. 
In this scenario, the user's best-response (i.e., naive) strategy is to ignore on any recommended video. 
This strategy is also their highest-payoff strategy because clicking on any video incurs a negative payoff. 
Therefore, $p$ is strategy-proof, 
but it does not lead to positive outcomes for the user or the platform. 

The second option seeks to model each user's payoff function $U$ so that it is possible to infer how the user would behave under alternative platform strategies $(p, \hcQ)$. 
While this approach may work in simple settings, it is difficult to do in complex settings. 
In particular, 
this approach becomes challenging when users are {heterogeneous}, 
there is {unobserved confounding}, 
and
the space of possible platform strategies $(p, \hcQ)$ is large.
Imagine, for instance,
predicting how an arbitrary social media user would behave under any possible feed. 
Such an estimation task is notoriously challenging. 
For one, both the user's and platforms action spaces are large (i.e., the user can interact with content in many possible ways, and there are many possible posts that the platform could  recommend). 
For another, 
each user behaves differently (i.e., there is high heterogeneity). 
In similarly complex (or high-risk) settings, 
developing reliable models is difficult. 

Third, 
one might naturally think to expand $\hcQ$ in order to address issues that arise from misspecification. 
We show in  Proposition \ref{prop:cf_omega}, however, that expanding the hypothesis class $\hcQ$ can,in fact, lower the platform's payoff and does not necessarily remove misspecification unless $q^\BR$ is guaranteed to be in the new hypothesis class. 
Since it is difficult (if not impossible) to determine whether the new hypothesis class contains $q^\BR$ \emph{a priori}, 
this approach should be taken with caution. 

Finally, 
we argue in the following section that
the fourth and fifth approaches described above can indeed improve outcomes for the user and platform, 
but are less straightforward than interventions that directly boost trustworthiness.

\subsubsection{Interventions that target trustworthiness}

In this section, 
we provide two simple recommendations that target trustworthiness. 
The goal of these interventions is to induce the user to behave naively (i.e., ``exogenously'').  
Recall from Definition \ref{def:S-strat-user}
that a user strategizes when the effort to deviate from their best-response (naive) behavior is low compared to the payoff gained by strategizing. 
At the core of the interventions below is the idea that, instead of increasing the effort required to deviate from best-response behavior (i.e., adding a barrier to strategization, as strategy-proof mechanisms often do), 
the platform can instead reduce the gains of strategization below the effort to deviate.

\paragraph{Recommendation \#1: Offering multiple algorithms.}
Offering users multiple algorithms from which they can choose addresses several issues simultaneously.
As an example, consider Twitter, 
which offers personalized, chronological, and trending feeds. 
Returning to the three reasons users strategize, as given in Section \ref{sec:reasons}, 
allowing users to select between multiple algorithms at each time step $t$ is a straightforward way of accommodating unobserved confounding, 
e.g., a hidden state $x_t$. 
While the platform could predict $x_t$, 
doing so is inevitably less reliable than giving users the ability to select an algorithm based on $x_t$ themselves. 
Second, 
if one algorithm is incompatible with a user's payoffs, 
the user will simply ignore that algorithm. 
On the other hand, 
if the platform only offers one algorithm $p$, 
it inevitably alienates users whose payoffs are incompatible with $p$. 

Offering multiple algorithms can therefore diminish two of the reasons that users strategize (as given in Section \ref{sec:reasons}) and guarantee that a user's limiting payoff is at least as high as under a single algorithm. It therefore improves trustworthiness by Definition \ref{def:trust}.

\paragraph{Recommendation \#2: Providing feedback mechanisms.}
Another intervention that builds trust and therefore mitigates the risks of strategization is providing users opportunities to give meaningful feedback. 
Returning to the reasons for strategization in Section \ref{sec:reasons},
feedback can serve as a simple and reliable indicator for misspecification. 
If a user consistently indicates that the platform is not behaving in their interest,
then the platform not only learns that this user is misspecified under $\hcQ$,
but also gains insight into how $\hcQ$ can be improved. 
If the feedback mechanism is designed such that giving feedback requires little effort \emph{and} the platform incorporates feedback well, 
then the user's payoff under best-response behavior is improved without much effort. As a result, strategization offers little benefit compared to the effort of deviating from $q^\BR$. 

There are many ways of eliciting feedback, 
and not all mechanisms are made equal. 
Platforms must ensure that they are comprehensive but not overwhelming, 
easily accessible but not so pervasive that users simply treat it as an annoyance.
Determining the precise design of feedback mechanisms is out of the scope of this work and may be of interest in future work.

\section{Acknowledgements}
The authors would like to thank Manolis Zampetakis for comments on 
an earlier version of this manuscript, and Hannah Li and Jennifer Allen 
for fruitful discussions.
Work supported in part by the NSF grants CNS-1815221 and DMS-2134108, and Open
Philanthropy.

\begin{refcontext}[sorting=nyt]
\printbibliography
\end{refcontext}

\clearpage
\appendix
\addcontentsline{toc}{section}{Appendix} %
\part{Appendix} %
\parttoc %

\clearpage
\section{Extended Related Work}
\label{app:full_related_work}
In this section, we provide an extended account of related work 
from \Cref{sec:rel_work}.
\paragraph{Endogenous learning.}
One of the key aspects of our model is {\em endogeneity} on the side of the 
platform, i.e., the platform's actions affect the data it 
collects. 
There is a vast literature in both economics and computer science 
that studies endogenous learning, some of which we heavily rely on in 
this work.

Our work draws from a long line of work at the intersection of economics and
game theory that explores the dynamics of {\em endogenous misspecified learning}.
Although we do not explicitly rely on it here,
the Berk-Nash equilibrium concept \citep{esponda2016berk}
provides a basis for much of this work.
The concept was later refined, 
expanded, and improved upon in several directions
\citep{Fudenberg2017-rb,Fudenberg2021-tu,bohren2016informational,Bohren2021-rs,frick2020stability}.
In this work, we rely most heavily on the results of \citet{frick2020stability}
due to their generality, 
but our results are portable in the sense that new ways of characterizing 
globally stable beliefs will allow for even sharper results in our setting.

There are also many related equilibrium concepts in the Economics literature
that we do not explore in this work, 
such as self-confirming equlibrium \citep{Fudenberg1993-xi} and subjective
equilbrium \citep{Kalai1993-ze, Kalai1995-kk}.
Relating these models (and others) to our setting is an interesting question 
for future work.

Several works in the computer science literature also study misspecified 
learning.
Of these, the most related is the work of \citet{Perdomo2020-za}
which introduces the idea of {\em performative prediction}.
The performative prediction setup mirrors that of supervised learning, 
except that the learner's current parameter estimate $\theta$
affects the data distribution $D_\theta$.
Our model can be viewed as an instance of performative prediction 
(which in turn, can be viewed as an instance of reinforcement learning
\citep{Brown2022-oz})
that focuses on a specific kind of performativity induced 
by users adapting to their platforms.

\vspace*{-1em}
\paragraph{Strategic classification and Stackelberg games.}
Strategic classification \citep{Hardt2016-ks} is a two-player game in which 
one player (the {\em decision-maker}) deploys a machine learning algorithm,
and the other player (the {\em decision subject}) strategically reports 
their features to attain a favorable decision.
Strategic classification also features endogenous learning, and is a 
special case of the performative prediction setup mentioned above.
Strategic classification has been the subject of a deluge of recent work in
computer science \citep{Bruckner2012-fs,Hardt2016-ks,Dong2018-rm,Ghalme2021-nf,Zrnic2021-vh}. 
Broadly, our model differs from one of strategic classification in that 
(a) we restrict the platform to a specific learning algorithm;
(b) users have their own utility functions that they can optimize arbitrarily,
and are not bound to small perturbations of some ``ground-truth'' features;
(c) users can be myopic or non-myopic in how they interact with the platform.

Prior works have studied deviations from these assumptions.
For (a), \citet{Zrnic2021-vh} study the case where rather than
reacting instantaneously, 
the decision subject is {\em learning} at the same time as the the
decision-maker.
Closer to our work, \citet{Levanon2022-zx} introduce {\em generalized} 
strategic classification, where the decision subject has a utility 
function that can be more aligned with that of the platform.
Finally, \citet{Haghtalab2022-bv} study Stackelberg games 
(which generalize strategic classification)
with {\em non-myopic agents} (i.e., where agents seek to optimimize
their long-term, rather than immediate, payoff).
As previously discussed, our model differs from theirs in that 
in our case users optimize their expected {\em limiting} payoff,
which, and so delaying user input has no bearing on the game dynamics.

\paragraph{User strategization on recommender systems.}
There are several related lines of work to ours that concern user strategization 
on recommender systems specifically. 
One line of work studies the ways in which users try to influence (and succeed
in influencing) their recommendations on popular platforms 
\citep{Siles2022-yv,Kim2023-tv,Simpson2022-vg}.
Closer to our work are theoretical models of user strategization on 
recommender platforms \citep{Cen2022-bc,Haupt2023-vo}, 
where \citep{Cen2022-bc} is an earlier version of this work.
Compared to the latter of these works, 
we propose a model that (a) extends more generally to data-driven platforms;
(b) allows us to study the causes and effects of strategization on the platform
and its ability to make counterfactual inferences;
and (c) connects strategization back to trustworthiness of the platform.

Another related work is that of \citet{Kleinberg2022-wy},
who study a model of {\em inconsistent preferences} under which 
users have a ``myopic'' system 1 that consumes 
addictive content and a ``non-myopic'' system 2 that considers
the value of content. 
Although we also study myopic and non-myopic users,
we use these terms in way that is subtly different from
\citep{Kleinberg2022-wy}, 
and their focus (inconsistent preferences) 
is neither necessary nor sufficient for strategization
to emerge.
Additionally, strategization concerns users' response 
to their platforms' algorithm, whereas in \citep{Kleinberg2022-wy}
the platform's algorithm is fixed.

\paragraph{Multi-agent learning and games.} 
There are also many lines of work on multi-agent learning, 
including multi-agent reinforcement learning 
(see \citep{Zhang2021-on, Tan1993-cq, Busoniu2008-lt} and their references),
multi-agent Bayesian learning (e.g., \citep{Wu2020-bl,Wu2021-hl}),
and inverse reinforcement learning (particularly its
cooperative \citep{Hadfield-Menell2016-mv} and adversarial \citep{Yu2019-oq} variants).
All of these models capture agents learning and interacting with each 
other at varying levels of generality. 
Here, we focus on a particular instantiation of such setups  
where two agents 
interact in a very specific way, mirroring the 
interaction between a user and a data-driven platform.
Our interests are also more specifically in studying the causes and effects of
strategic behavior in this setup.

Another related work studies repeated alternating two-player games
\citep{Roth2010-yk} and considers the complexity of computing  
equilibria in such games.
In our work, we avoid having to compute such equilibria by assuming that 
the platform follows a fixed Bayesian updating strategy.

\paragraph{Mechanism design and incentive-compatible auctions.}
Our notion of ``trustworthiness'' from \cref{sec:trust_def} 
is highly related to the idea of incentive-compatibility in 
mechanism design. There is a rich line of work in econometrics on 
designing incentive-compatible methods for repeated games in general
(see, e.g., \citep{Mailath2006-oy, Bergin1993-ly, Kalai1993-ze}),
and more specifically for repeated auctions
\citep{Deng2021-mn,Abernethy2019-js,Nedelec2022-xs,Amin2013-lc,Kanoria2019-ri}.
In the latter, a long line of work stems from the tight connections between 
incentive-compatibility and differential privacy 
\citep{McSherry2007-mz, Nissim2012-oe}.

\paragraph{Human-AI collaboration.} 
Finally, our work contributes to a growing body of work on 
Human-AI collaboration \citep{Wang2019-so, Wang2020-ty, Mozannar2022-sd} and teaming 
\citep{Zhang2021-gc, Bansal2019-jh},
which both study primarily collaborative interactions 
between AI systems and their users. 
In particular, one can view 
our observations about conditions for strategization
(and corresponding recommendations about trustworthy algorithm 
design) in \cref{sec:elements_trust} as principles for human-AI 
collaboration in the case where the AI and the human 
are not entirely aligned.

Within Human-AI collaboration, one line of work in particular 
isolates {\em trust} as a building block of successful interaction
(see, e.g., \citep{Okamura2020-eu, Hou2023-cb, Ezer2019-zl, Bao2021-fn}
and references therein).
Our work supports the high-level idea that designing systems with trust in mind
is important, but we explore a slightly different notion of trust than what is
typically considered.
In particular, 
in our case trustworthiness dictates whether a user will maximize their 
immediate utility at each step of a game.
By comparison, human-AI collaboration usually explores trust to the end 
of getting users to use AI systems in the first place (or comfortably 
delegate complext tasks to AI systems).

\clearpage
\section{Omitted Proofs: Illustrative Example}
In this section, we provide the proofs omitted from \cref{sec:example}.

\subsection{Proof of \cref{prop:limiting_illustrative}}
\label{app:limiting_illustrative_pf}
\illustrativestable*
\begin{proof}
	In this proof, we will use the result of \citet{frick2020stability}, 
	restated as \cref{thm:global_stability}.
	Note that the theorem cannot be applied directly due to the 
	``bounded likelihood ratios'' regularity
	condition being violated---still, since $p(\cdot;\mu)$ has full support,
	we can always wait until the violating user models are eliminated by 
	Bayes' rule (which happens almost surely), 
	and then apply the results from \citep{frick2020stability}
	thereafter. 

	For example, note that $\hq_3(B|Z) > 0$ for all $B \in \cB$ and $Z \in \cZ$,
	so if $\supp(q) \cap \cZ_{A} \neq \emptyset$, then with probability one we 
	will eventually see a $(B, Z)$ such that $B = 1$ and $Z \in \cZ_A$, at 
	which point
	\begin{align*}
		\mu(\hq_2) = \frac{
			\mu_{t-1}(\hq_2) \cdot \hq_2(B_t|Z_t)
		}{
			\sum_{\hq_i \in \hcQ} \mu_{t-1}(\hq_i) \cdot \hq_i(B_t|Z_t)
		} = 0,
	\end{align*}
	and thus we can apply the results from \citep{frick2020stability}
	to the set $\{\hq_1, \hq_3\}$.
	
	note that for {\em any} belief $\mu$ and any user model $\hq$,
	\begin{align*}
		KL(p^\mu \times q, p^\mu \times \hq) &= \mathbb{E}_{(B, Z) \sim p^\mu \times q}\left[
			\log\left(
				\frac{p^\mu(Z) \cdot q(B|Z)}{p^\mu(Z) \cdot \hq(B|Z)}
			\right)
		\right] \\
		&= \mathbb{E}_{Z \sim p^\mu}\left[\mathbb{E}_{B \sim q(\cdot|Z)}\left[
			\log\left(
				\frac{q(B|Z)}{\hq(B|Z)}
			\right)
		\right]\right],
	\end{align*}
	and so for the user model $\hq_3$, $\varepsilon \leq \hq_3(B|Z) \leq
	1-\varepsilon$ for all $B \in \cB, Z \in \cZ$,
	and so the quantity above is guaranteed to be finite.

	Now, since the algorithm $p_0$ recommends a random item $Z$ with probability $\varepsilon$,
	\begin{align*}
		KL(p^\mu \times q, p^\mu \times \hq) &\geq \varepsilon \cdot \frac{1}{|\cZ|} \sum_{Z \in \cZ} \mathbb{E}_{B \sim q(\cdot|Z)}\left[
			\log\left(
				\frac{q(B|Z)}{\hq(B|Z)}
			\right)
		\right] \\
		&=  \varepsilon \cdot \frac{1}{|\cZ|} \left(\sum_{Z \in \cZ_A} \mathbb{E}_{B \sim q(\cdot|Z)}\left[
			\log\left(
				\frac{q(B|Z)}{\hq(B|Z)}
			\right)
		\right] + \sum_{Z \in \cZ_B} \mathbb{E}_{B \sim q(\cdot|Z)}\left[
			\log\left(
				\frac{q(B|Z)}{\hq(B|Z)}
			\right)
		\right]\right) \\
		&\geq  \varepsilon \cdot \frac{1}{|\cZ|} \left(\max_{Z \in \cZ_A} q(B=1|Z)\cdot \log\left( \frac{q(B=1|Z)}{\hq(B=1|Z)} \right) \right. \\ 
		&\qquad\qquad \qquad + \left.\max_{Z \in \cZ_B} q(B=1|Z) \cdot\log\left( \frac{q(B|Z)}{\hq(B|Z)} \right)\right).
	\end{align*}
	Now, if $|\text{supp}(q) \cap \cZ_A| > 0$, then the first maximum is infinite for $\hq = \hq_2$, 
	since there will be at least one element $Z \in \cZ_A$ for which $q(B=1|Z) > 0$,
	but $\hq_2(B=1|Z) = 0$ for $Z \in \cZ_A$.
	Similarly, if $|\text{supp}(q) \cap \cZ_B| > 0$, the second maximum will be infinity for $\hq_1$.
	This observation suffices to show that if both $|\text{supp}(q) \cap \cZ_B| > 0$ and $|\text{supp}(q) \cap \cZ_B| > 0$, 
	the platform will converge to $\hq_3$.

	To complete the proof, suppose without loss of generality that $|\text{supp}(q) \cap \cZ_B| = 0$.
	We need to prove that in this case, the platform will converge to $\hq_1$ (and in particular, not $\hq_3$).
	We will show this by using KL-dominance---in particular, for any belief $\mu$,
	\begin{align*}
		KL(p^\mu \times q, p^\mu \times \hq_3) - KL(p^\mu \times q, p^\mu \times \hq_1)
		&= \mathbb{E}_{(B, Z) \sim p^\mu \times q}\left[
			\log\left(
				\frac{q(B|Z)}{\hq_3(B|Z)}
			\right)
			- 
			\log\left(
				\frac{q(B|Z)}{\hq_1(B|Z)}
			\right)
		\right]  \\ 
		&=  \mathbb{E}_{(B, Z) \sim p^\mu \times q}\left[
			\log\left(
				\frac{\hq_1(B|Z)}{\hq_3(B|Z)}
			\right)
		\right]  \\
		&=  \mathbb{E}_{(B, Z) \sim p^\mu \times q}\left[
			\log\left(
				\frac{\hq_1(B|Z)}{\hq_3(B|Z)}
			\right)
		\right].
	\end{align*}
	Since $\hq_3(\cdot|Z) = \hq_1(\cdot|Z)$ for $Z \in \cZ_A$, this simplifies to
	\begin{align*}
		&=  \mathbb{E}_{Z \sim p^\mu}\left[
			\mathbb{E}_{B \sim q(\cdot|Z)}\left[
			\log\left(
				\frac{\hq_1(B|Z)}{\hq_3(B|Z)}
			\right)\right]
			\bigg|Z \in \cZ_B
		\right] \cdot \mathbb{P}\left(
			Z \in \cZ_B
		\right).
	\end{align*}
	By assumption (i.e., that $|\text{supp}(q) \cap \cZ_B| = 0$), we know that $q(\cdot|Z) = \delta_{B=0}$, and so
	\begin{align*}
		&=  \mathbb{E}_{Z \sim p^\mu}\left[
			\log\left(
				\frac{\hq_1(B=0|Z)}{\hq_3(B=0|Z)}
			\right)
			\bigg|Z \in \cZ_B
		\right] \cdot \mathbb{P}\left(
			Z \in \cZ_B
		\right) \\
		&=  \log\left(
				\frac{1}{\gamma}
			\right) \cdot \mathbb{P}\left(
			Z \in \cZ_B
		\right) \\
		&> 0.
	\end{align*}
	Thus, $\hq_1$ strictly dominates $\hq_3$ at all beliefs $\mu$, and so the platform will converge to $\hq_1$.
\end{proof}

\subsection{Proof of \cref{prop:illustrative_strat_happens}}
\label{app:illustrative_strat_happens_proof}
\illustrativestrathappens*
\begin{proof}
	We begin with the reverse direction.
	Suppose $\cZ^+$ is not 
	fully contained in either $\cZ_A$ or $\cZ_B$. By \cref{prop:limiting_illustrative}.
	If the user is naive, the platform will converge to the belief $\mu(\hq_3) = 1$, 
	and so the limiting proposition distribution will be a uniform distribution over all items
	(since $\hq_3(B=1|Z) = 1-\varepsilon$ for all $Z$).
	The expected user payoff will then be equal to
	\[
		\brU(p^{\delta_{\hq_3}}, q^\BR) = \frac{|\cZ^+|}{|\cZ|}.
	\]
	Without loss of generality, suppose $\cZ_A$ is a better approximation than $\cZ_B$ 
	to $\cZ_1$, and in particular
	\[
		\Delta \coloneqq \frac{|\cZ_A \cap \cZ^+|}{|\cZ_A|} - \frac{|\cZ_B \cap \cZ^+|}{|\cZ_B|} > 0.
	\]
	We define $\eps$ (the probability with which $p_0$ recommends items uniformly at random) as
	\[
		\eps < \frac{|\cZ_A|\cdot |\cZ_B| \cdot \Delta}{|\cZ_A \cap \cZ^+|}.
	\]
	Now, rearranging this definition yields
	\begin{align*}
		|\cZ_A|\cdot |\cZ_B| \cdot \left(
			\frac{|\cZ_A \cap \cZ^+|}{|\cZ_A|} - \frac{|\cZ_B \cap \cZ^+|}{|\cZ_B|}
		\right) &> |\cZ_A \cap \cZ^+|\cdot \eps \\
		{|\cZ_B| \cdot |\cZ_A \cap \cZ^+|} - {|\cZ_A| \cdot |\cZ_B \cap \cZ^+|} &> |\cZ_A \cap \cZ^+|\cdot \eps \\
		(1-\varepsilon) \cdot {|\cZ_B| \cdot |\cZ_A \cap \cZ^+|} &> \eps \cdot |\cZ_A \cap \cZ^+|\cdot |\cZ_A| + {|\cZ_A| \cdot |\cZ_B \cap \cZ^+|}  \\
		(1-\eps)\cdot |\cZ_A \cap \cZ^+|\cdot |\cZ_A| + (1-\varepsilon) \cdot {|\cZ_B| \cdot |\cZ_A \cap \cZ^+|} &>  |\cZ_A \cap \cZ^+|\cdot |\cZ_A| + {|\cZ_A| \cdot |\cZ_B \cap \cZ^+|}  \\
		(1-\eps)\cdot |\cZ_A \cap \cZ^+|\cdot |\cZ| &>  |\cZ^+|\cdot |\cZ_A|  \\
		\frac{|\cZ_A \cap \cZ^+|}{|\cZ_A|} &> \frac{|\cZ^+|}{|\cZ|} \cdot \frac{1}{1-\eps}.
	\end{align*}
	Now, consider an alternative strategy $q'$ where the user clicks only on items in
	$\cZ_+ \cap \cZ_A$. By \cref{prop:limiting_illustrative}, this would lead the 
	platform to converge to $\hq_1$, and so the user's payoff would be 
\begin{align*}
	\brU(p^{\delta_{\hq_1}}, q') = \varepsilon \cdot \frac{|\cZ^+ \cap \cZ_A|}{|\cZ|} + (1-\varepsilon)\cdot\frac{|\cZ^+ \cap \cZ_A|}{|\cZ_A|} \\
	> \varepsilon \cdot \frac{|\cZ^+ \cap \cZ_A|}{|\cZ|} + \frac{|\cZ^+|}{|\cZ|},
\end{align*}
and so playing the strategy $q'$ guarantees the user strictly higher long-run payoff than $q^\BR$.

For the forward direction, suppose without loss of generality that $\cZ^+ \subset \cZ_A$.
Observe that conditioned on a fixed distribution of propositions $p^\mu \in \Delta(Z)$
(i.e., in the absence of any learning), 
$q^\BR$ by definition yields optimal expected payoff $\brU(p^\mu, q)$.
As a result, $\brU(p^\mu, q^\BR)$ gives an upper bound for the user's payoff when 
the platform's belief is $\mu$.
Now, if the platform belief is $\mu(\hq_3) = 1$, then the resulting 
distribution is a uniform distribution over items $Z$, and so
\begin{align*}
	\max_{q \in \cQ} \brU(p^{\delta_{\hq_3}}, q) \leq \brU(p^{\delta_{\hq_3}}, q^\BR) = \frac{|\cZ^+|}{|\cZ|} = \frac{|\cZ^+|}{|\cZ_A| + |\cZ_B|}.
\end{align*}
Similarly, if the platform believes $\mu(\hq_2) = 1$, then it will recommend uniformly with probability $\eps$,
and otherwise recommend from $\cZ_B$, and so
\begin{align*}
	\max_{q \in \cQ} \brU(p^{\delta_{\hq_2}}, q) \leq \brU(p^{\delta_{\hq_2}}, q^\BR) = \eps \cdot \frac{|\cZ^+|}{|\cZ|} + (1-\eps) \cdot 0 = \eps \cdot \frac{|\cZ^+|}{|\cZ|}.
\end{align*}
Now, by \cref{prop:limiting_illustrative}, playing any non-degenerate strategy (i.e., $q$ for which $\max_Z q(B=1|Z) > 0$) 
will lead the platform to converge to a point mass belief. The naive strategy $q^\BR$ will lead to 
the platform having limiting belief $\hq_1$, for which the user's payoff is 
\begin{align*}
	\brU(p^{\delta_{\hq_1}}, q^\BR) = \eps \cdot \frac{|\cZ^+|}{|\cZ|} + (1-\eps) \cdot \frac{|\cZ^+|}{|\cZ_A|}.
\end{align*}
This is clearly greater than the other two upper bounds, and so if the user plays a non-degenerate strategy,
$q^\BR$ is the best option. The proof concludes by noting that the degenerate strategy yields a utility of zero.
\end{proof}

\subsection{Proof of \cref{prop:illustrative_improve_utility}}
\label{app:illustrative_improve_utility_proof}
\illustrativeimproveutility*
\begin{proof}
	Let $\mu_S$ be the platform's long-run belief when the user is strategic, and 
	let $\mu_\BR$ be the platform's long-run belief when the user is naive. Note 
	that \cref{prop:limiting_illustrative} implies that these beliefs exist and 
	are unique for $q^\BR$ and for any non-degenerate $q^*$. 
	Now, if the user is incentivized to strategize, it must be that 
	$\brU(p^{\mu_\BR}, q^\BR) \leq \brU(p^{\mu_S}, q^*)$. In other words,
	\begin{align*}
		\sum_{Z \in \cZ} p^{\mu_S}(Z) \bbE_{q^S(\cdot|Z)}[U(Z, B)] 
		&\geq
		\sum_{Z \in \cZ} p^{\mu_\BR}(Z) \bbE_{q^\BR(\cdot|Z)}[U(Z, B)] \\ 
		\sum_{Z \in \cZ} p^{\mu_S}(Z) q^S(B=1|Z) a(Z)
		&\geq 
		\sum_{Z \in \cZ} p^{\mu_\BR}(Z) q^\BR(B = 1|Z) \cdot a(Z) \\
		\sum_{Z \in \cZ^+} p^{\mu_S}(Z) q^S(B=1|Z)
		-
		\sum_{Z \in \cZ \setminus \cZ^+} p^{\mu_S}(Z) q^S(B=1|Z)
		&\geq
		\sum_{Z \in \cZ^+} p^{\mu_\BR}(Z) q^\BR(B = 1|Z) \\
		\sum_{Z \in \cZ} p^{\mu_S}(Z) q^S(B=1|Z)
		-
		\sum_{Z \in \cZ} p^{\mu_\BR}(Z) q^\BR(B = 1|Z)
		&\geq 
		2\sum_{Z \in \cZ \setminus \cZ^+} p^{\mu_S}(Z) q^S(B=1|Z)
	\end{align*}
\end{proof}

\subsection{Proof of \cref{prop:illustrative_cfx_bad}}
\label{app:illustrative_cfx_bad_proof}
\illustrativecfxbad*
\begin{proof}
	For convenience, assume $|\cZ_B|$ is even (otherwise, the same proof holds but requires 
	us to keep track of a rounding error).
	We will define $a(Z)$ to be $+1$ on a half of $\cZ_B$ and $-1$ on the other half,
	and $+1$ on $3/4$ of $\cZ_A$ (and $-1$ on the remaining $1/4$).
	By construction, defining $\cZ^+ = \{Z \in \cZ: a(Z) = 1\}$,
	$$\frac{|\cZ_A \cap \cZ^+|}{|\cZ_A|} \geq \frac{3}{4} > \frac{1}{2} = \frac{|\cZ_B \cap \cZ^+|}{|\cZ_B|}.$$
	Using an identical argument to the one used in \cref{app:illustrative_strat_happens_proof},
	we see that when $\eps$ is sufficiently small, 
	a strategic user will restrict their clicks to $\cZ^+ \cap \cZ_1$,
	in order to induce a platform belief of $\hq_1$.
	For some small $\alpha > 0$ (to be defined later), 
	consider the toxicity function
	\begin{align*}
		\textsc{toxicity}(Z) = \begin{cases}
			\alpha &\text{ if } Z \in \cZ_A \cap \cZ^+ \text{ or } \cZ_B \setminus \cZ^+ \\
			1 &\text{otherwise}.
		\end{cases}
	\end{align*}
	In the remainder of the proof, we 
	introduce the following notation for convenience:
	\begin{align*}
		n_1 = |\cZ_A \cap \cZ^+|, 
		\qquad n_2 = |\cZ_A \setminus \cZ^+|
		\qquad n_3 = |\cZ_B \cap \cZ^+|,
		\qquad n_4 = |\cZ_B \setminus \cZ^+|.
	\end{align*}
	\paragraph{Platform's current payoff.} The platform's current payoff under the user's strategic
	behavior is 
	\begin{align}
		\label{eq:current_payoff_prf}
		\brV^*(p_0, \hcQ) = \eps \cdot \frac{n_1}{n_1 + n_2 + n_3 + n_4} + (1-\eps) \cdot \frac{n_1}{n_1 + n_2}.
	\end{align}
	Note that the platform's predicted payoff under $p_0$ is given by 
	\begin{align*}
		\hV(p_0, \hcQ) = (1-\gamma)\left[
			\eps \cdot \frac{n_1 + n_2}{n_1 + n_2 + n_3 + n_4} + (1-\eps),
		\right],
	\end{align*}
	and by setting $\gamma$ appropriately we can make these two quantities equal 
	(intuitively, since $\eps$ is small, 
	$1-\gamma$ roughly corresponds to the fraction of $\cZ_A$
	that the user actually engages with, i.e., $\frac{n_1}{n_1 + n_2}$).
	\paragraph{Platform's predicted payoff.} We now derive the platform's predicted payoff 
	$\widehat{V}(p_1, \hcQ)$---the platform believes that the user's strategy is $\hq_1$, and
	so its predicted payoff is 
	\begin{align*}
		\widehat{V}(p_1, \hcQ) 
		&= (1 - \gamma)\left[\eps \cdot \frac{\alpha n_1 + n_2}{\alpha n_1 + n_2 + n_3 + \alpha n_4} + (1 - \eps)\right].
	\end{align*}
	From this, it is straightforward to show that
	\begin{align*}
		\widehat{V}(p_1, \hcQ) - \brV^*(p_0, \hcQ) 
		&= \widehat{V}(p_1, \hcQ) - \widehat{V}(p_0, \hcQ) \\
		&= 
		(1 - \gamma)\cdot \eps \cdot \left[\frac{\alpha n_1 + n_2}{\alpha n_1 + n_2 + n_3 + \alpha n_4} 
											- \frac{n_1 + n_2}{n_1 + n_2 + n_3 + n_4} \right] \\
		&< 0,
	\end{align*}
	where in the last inequality we use the fact that:
	\begin{align*}
		\frac{n_2}{n_1 + n_2} &\leq \frac{1}{4} < \frac{1}{2} = \frac{n_3}{n_3 + n_4} \\
		n_2n_3 + n_2 n_4 &< n_1n_3 + n_2n_3 \\
		0 &> (1-\alpha)(n_2 n_4 - n_1n_3)  \\
		&= (\alpha n_1 + n_2)(n_1 + n_2 + n_3 + n_4) - (n_1 + n_2)(\alpha n_1 + n_2 + n_3 + \alpha n_4).
	\end{align*}
	This sequence of calculations proves the result (b).
	\paragraph{Platform's true payoff.}
	We now consider the platform's true payoff when the user is strategic.
	Now, we use the fact that (a) if the user's behavior is non-degenerate, 
	the platform's belief will converge to a single $\mu(\hq_i) = 1$; 
	(b) for a fixed belief $\mu$, the maximum user payoff is upper bounded 
	by the payoff attained by $q^\BR$. 
	
	Thus, under the toxicity function above,
	\begin{align*}
		\max_{q \in \cQ} \brU(p_1^{\delta_{\hq_1}}, q) \leq \brU(p_1^{\delta_{\hq_1}}, q^\BR) &= 
		\eps \cdot \frac{\alpha n_1 + n_3}{\alpha n_1 + n_2 + n_3 + \alpha n_4} + (1- \eps) \cdot \frac{\alpha n_1}{\alpha n_1 + n_2}\\
		\max_{q \in \cQ} \brU(p_1^{\delta_{\hq_2}}, q) \leq \brU(p_1^{\delta_{\hq_2}}, q^\BR) &= 
		\eps \cdot \frac{\alpha n_1 + n_3}{\alpha n_1 + n_2 + n_3 + \alpha n_4} + (1- \eps) \cdot \frac{n_3}{n_3 + \alpha n_4}\\
		\max_{q \in \cQ} \brU(p_1^{\delta_{\hq_3}}, q) \leq \brU(p_1^{\delta_{\hq_3}}, q^\BR) &= \frac{\alpha n_1 + n_3}{\alpha n_1 + n_2 + n_3 + \alpha n_4}.
	\end{align*}
	Observe that as $\alpha \to 0$,
	\begin{align}
		\label{eq:upper_bounds}
		\max_{q \in \cQ} \brU(p_1^{\delta_{\hq_1}}, q) \to \frac{\eps\cdot n_3}{n_2 + n_3} \qquad \text{ and } \qquad 
		\max_{q \in \cQ} \brU(p_1^{\delta_{\hq_3}}, q) \to \frac{n_3}{n_2 + n_3}.
	\end{align}
	Also, if the user plays the strategy $q^\dagger(B=1|Z) = \ind{Z \in \cZ^+ \cap \cZ_B}$,
	the platform will converge to $\mu(\hq_2) = 1$ (by \cref{prop:limiting_illustrative}),
	and so the user's utility will be 
	\begin{align*}
		\brU(p_1^{\delta_{\hq_2}}, q^\dagger) &= 
		\eps \cdot \frac{n_3}{\alpha n_1 + n_2 + n_3 + \alpha n_4} + (1- \eps) \cdot \frac{n_3}{n_3 + \alpha n_4},
	\end{align*}
	which as $\alpha \to 0$, converges to 
	\begin{align*}
		\brU(p_1^{\delta_{\hq_2}}, q^\dagger) &\to
		\eps \cdot \frac{n_3}{n_2 + n_3} + (1- \eps)  > \frac{n_3}{n_2 + n_3} = \lim_{\alpha \to 0} \max_{q \in \cQ} \brU(p_1^{\delta_{\hq_3}}, q).
	\end{align*}
	In particular, by comparing this to \eqref{eq:upper_bounds}, there must exist some $\alpha > 0$ such that 
	\[
		\brU(p_1^{\delta_{\hq_2}}, q^\dagger) > \max_{q \in \cQ} \brU(p_1^{\delta_{\hq_3}}, q) > \max_{q \in \cQ} \brU(p_1^{\delta_{\hq_1}}, q).
	\]
	This implies that the strategic user will induce the belief $\mu(\hq_2)=1$ when the platform 
	plays the algorithm $p_1$. Note that of the user strategies that induce $\mu(\hq_2) = 1$, 
	the optimal one is clearly $q^\dagger$, as any strategy that does not set $B=1$ 
	for $Z \in \cZ^+ \cap \cZ_B$ would increase its utility by doing so, and any strategy 
	that does set $B=1$ for $Z \in \cZ_B \setminus \cZ^+$ would needlessly incur a penalty.

	Thus, to conclude the proof, observe that 
	\begin{align*}
		\brV^*(p_1, \hcQ) 
		&= \brV(p_1^{\delta_{\hq_2}}, q^\dagger) \\
		&= \eps \cdot \frac{n_3}{\alpha n_1 + n_2 + n_3 + \alpha n_4} + (1- \eps) \cdot \frac{n_3}{n_3 + \alpha n_4},
	\end{align*}
	which, as $\alpha \to 0$ and $\eps \to 0$, converges to $1$. Contrasting this to \eqref{eq:current_payoff_prf}
	makes it clear that by choosing $\alpha$ and $\eps$ small enough, we get that $\brV^*(p_1, \hcQ) > \brV^*(p_0, \hcQ)$.
\end{proof}

\subsection{Proof of \cref{prop:illustrative_expand_omega}}
\label{app:illustrative_expand_omega_proof}
\illustrativeexpandomega*
\begin{proof}
	Define the affinity function $a(Z)$ to be $+1$ on half of the items in $\cZ_A$
	(rounding down if $|\cZ_A|$ is not even),
	and on a single item from $\cZ_B$.
	
	Recall from \cref{prop:illustrative_strat_happens} and its proof in
	\cref{app:illustrative_strat_happens_proof} that a strategic 
	user will try to induce the belief $\mu(\hq_1) = 1$ from the platform by 
	restricting their clicks to $\cZ_A \cap \cZ^+$ (where recall that $\cZ^+$ is the 
	subset of $\cZ$ on which $a(Z) = 1$). For some $\eta > 0$ to be defined later,
	define 
	\[
		\hq_4(B=1|Z) = 1-\eta\quad \forall\ Z \in \cZ
	\]
	Now, consider a strategic user in reponse to the platform strategy $(p_0, \hcQ \cup \{\hq_4\})$.
	The user is faced with a choice between three cases:
	\begin{itemize}
		\item[(A)] 
		Force
		the platform to converge to $\hq_3$ or $\hq_4$. These are indistinguishable
		from the user's perspective as 
		in either case, the resulting proposition distribution $p_0^\mu$ will
		converge to a uniform distribution over $\cZ$.
		\item[(B)] Induce the platform to converge to $\mu(\hq_1) = 1$, in which case 
		the limiting proposition distribution is a mixture of the uniform distribution
		over $\cZ$ (with probability $\eps$) and the uniform distribution over $\cZ_A$
		(with probability $1-\eps$).
		\item[(C)] Conversely, induce the platform to converge to $\mu(\hq_2) = 1$.
	\end{itemize}
	Note that playing naively results in case (A), 
	and so the user's utility in that case is both upper and lower bounded by
	$|\cZ^+|/|\cZ|$. 
	Meanwhile, since we constructed $|\cZ_B \cap \cZ^+| = 1$, 
	the user's utility in case (C) is upper bounded by $1/|\cZ_B|$.
	We can thus remove case (C) from consideration.

	It thus remains to bound the utility of Case (B). In the absence of $\hq_4$, 
	the user can restrict their clicks to $\cZ_A \cap \cZ^+$ and guarantee that
	the platform converges to $\hq_1$ (see \cref{app:limiting_illustrative_pf}).
	In the presence of $\hq_4$, however, we argue that the only strategy that 
	guarantees the platform converging to $\hq_1$ will also result in low user payoff.

	To show this, suppose there exists a strategy $q$ for which the platform
	converges to $\hq_1$. If $\eta < \gamma$,
	\begin{align*}
		&\KL(q||\hq_1) - \KL(q||\hq_4) \\
		&= \mathbb{E}_{Z \sim p^\mu}\left[
			\mathbb{E}_{B \sim q(\cdot|Z)}\left[
			\log\left(
				\frac{q(B|Z)}{\hq_1(B|Z)}
			\right)
			-
			\log\left(
				\frac{q(B|Z)}{\hq_4(B|Z)}
			\right)
			\right]
		\right] \\
		&= \mathbb{E}_{Z \sim p^\mu}\left[
			\mathbb{E}_{B \sim q(\cdot|Z)}\left[
			\log\left(
				\frac{\hq_4(B|Z)}{\hq_1(B|Z)}
			\right)
			\right]
		\right] \\
		&= \mathbb{E}_{Z \sim p^\mu}\left[
			\mathbb{E}_{B \sim q(\cdot|Z)}\left[
			\log\left(
				\frac{\hq_4(B|Z)}{\hq_1(B|Z)}
			\right)
			\right] \bigg|
			Z \in \cZ_A
		\right] \cdot \mathbb{P}_{Z \sim p^\mu}(Z \in \cZ_A) \\
		&\qquad + \mathbb{E}_{Z \sim p^\mu}\left[
			\mathbb{E}_{B \sim q(\cdot|Z)}\left[
			\log\left(
				\frac{\hq_4(B|Z)}{\hq_1(B|Z)}
			\right)
			\right] \bigg|
			Z \in \cZ_B
		\right]  \cdot \mathbb{P}_{Z \sim p^\mu}(Z \in \cZ_B) \\
		&= \mathbb{E}_{Z \sim p^\mu}\left[
			q(B=1|Z) \cdot \log\left(\frac{1-\eta}{1-\gamma}\right)
			+
			q(B=0|Z) \cdot \log\left(\frac{\eta}{\gamma}\right)
			\bigg|
			Z \in \cZ_A
		\right] \cdot \mathbb{P}_{Z \sim p^\mu}(Z \in \cZ_A) \\
		&\qquad + \mathbb{E}_{Z \sim p^\mu}\left[
			\log\left(\eta\right)
			\bigg|
			Z \in \cZ_B
		\right]  \cdot \mathbb{P}_{Z \sim p^\mu}(Z \in \cZ_B) \\
		&\geq \left( \log\left(1-\eta\right)
			+
		    \mathbb{E}_{Z \sim p^\mu}\left[ q(B=0|Z) \big| Z \in \cZ_A \right]
			\cdot \log\left(\frac{\eta}{\gamma}\right)\right) 
			 \cdot \mathbb{P}(Z \in \cZ_A) + \log\left(\eta\right).
	\end{align*}
	Thus, if 
	\begin{align*}
		\mathbb{E}_{Z \sim p^\mu}\left[ q(B=0|Z) \big| Z \in \cZ_A \right] >
		\frac{\left(\frac{-\log(\eta)}{\mathbb{P}(Z \in \cZ_A)} - \log\left(
			{1-\eta}
		\right)\right)}{\log\left(\frac{\eta}{\gamma}\right)},
	\end{align*}
	then $\KL(q||\hq_1) - \KL(q||\hq_4) > 0$, and $\hq_4$ thus dominates $\hq_1$ at belief $\mu$.
	Now, since $\bbP(Z \in \cZ_A)$ is lower bounded by $\eps |\cZ_A|/|\cZ|$, 
	we can set $\gamma$ small so that the above expression evaluates to
	$\frac{1}{|\cZ|}$. 
	
	Also, note that since none of the user models distinguish between 
	different elements within the partition $\cZ_A$, the expectation on the left is equivalent to
	the unconditional expectation $\mathbb{E}_{Z \sim \text{Unif}(\cZ_A)}[\cdot]$.

	Putting these two observations together, we have that if 
	\begin{align}
		\label{eq:random_eq_1}
		\frac{1}{|\cZ_A|} \sum_{Z \in \cZ_A} q(B=0|Z) > \frac{1}{|\cZ|},
	\end{align}
	then $\hq_4$ strictly dominates $\hq_1$ at all beliefs $\mu$, and the platform will 
	converge to $\mu(\hq_4) = 1$. Since we have by assumption that the platform converges 
	to $\mu(\hq_1) = 1$, it must be that \eqref{eq:random_eq_1} is false, and so
	\begin{align*}
		\frac{1}{|\cZ_A|} \sum_{Z \in \cZ_A} q(B=1|Z) \geq 1 - \frac{1}{|\cZ|}.
	\end{align*}
	By construction ($\cZ^+$ containing exactly half of $\cZ_A$),
	\begin{align*}
		\frac{1}{|\cZ_A \cap \cZ^+|} \sum_{Z \in \cZ_A \cap \cZ^+} q(B=1|Z) \geq 1 - \frac{2}{|\cZ|}.
	\end{align*}
	We now argue that with such a strategy, the user can attain no more than $\frac{2}{|\cZ|}$ utility.
	Since none of the user models distinguish between elements within $\cZ_A$, a
	coupling argument shows that for every element $Z: a(Z) = 1$ that the user
	clicks on with probability $\delta$,
	there is at least one other element with $a(Z') = -1$ that the user 
	clicks on with probability $\delta(1 - \frac{2}{|\cZ|}$) and that is 
	recommended with equal probability to $Z$.

	We have thus shown that the maximum attainable payoff from case (B) is $2/|\cZ|$,
	which is lower than the guaranteed payoff from case (A) of $|\cZ^+|/|\cZ|$. 
	A strategic user will thus choose to play naively, which in fact {\em lowers}
	platform payoff.
	We have thus shown that both case (B) and case (C) result in low payoffs for the user,
	and so the user will choose case (A). The proof concludes by observing that platform's 
	payoff under case (A) decreases from the original strategic user.
\end{proof}

\clearpage
\section{Omitted Proofs: Main Results}

\subsection{Proof of \cref{prop:strat_helps_platform}}
\label{app:strat_helps_proof}
\strathelps*
\begin{proof}
Note that the left-hand side of 
\eqref{eq:strategization_helps_platform2}
is exactly the platform's payoff under a strategic user,
and so it only remains to show that the right-hand side 
corresponds to the platform's payoff under a naive user.
First, by the min-max inequality,
\begin{align*}
	\max_{q \in \cQ} \min_{\mu \in \Delta(\hcQ)}
	\overline{U}(\pmu{\mu}{\cdot}, q)
	&\leq \min_{\mu \in \Delta(\hcQ)} \max_{q \in \cQ} 
	\overline{U}(\pmu{\mu}{\cdot}, q) \\
	&\leq \min_{\mu \in \Delta(\hcQ)} 
	\overline{U}(\pmu{\mu}{\cdot}, q^\BR),
\end{align*}
and so 
\[
	q^\BR \in \arg\max_{q \in \cQ} \min_{\mu \in \Delta(\hcQ)}
	\overline{U}(\pmu{\mu}{\cdot}, q).
\]
Now, suppose $q \neq q^\BR$ satisfies
\[
	q \in \arg\max_{q \in \cQ} \min_{\mu \in \Delta(\hcQ)}
	\overline{U}(\pmu{\cdot}{\mu}, q).
\]
Since $q \neq q^\BR$ and $U$ has a unique maximizer $B^*(Z)$ for 
each $Z \in \cZ$, there must exist some $Z_0 \in \cZ$ for which 
$q(\cdot|Z_0) \neq \ind{B^*(Z_0)}$.
Define 
\[
	q'(\cdot|Z) = \begin{cases}
		q(\cdot|Z) &\text{ if $Z \neq Z_0$} \\
		\ind{B^*(Z_0)} &\text{ otherwise.}
	\end{cases}
\]
Clearly, $q'$ attains the same payoff $U$ as $q$ for any 
$Z \neq Z_0$, and when $Z = Z_0$, $q'$ attains better payoff 
than $q$. If there exists a $\mu$ such that $p(Z_0;\mu) > 0$ 
we have thus reached a contradiction, and otherwise $Z_0$ 
will never by played by the platform, and so user 
actions under $q$ are indistinguishable to the platform 
from user actions under $q^\BR$.

\end{proof}

\subsection{Proof of \cref{prop:cf_p}}
\label{app:nonsmooth_proof}
\badpredictiveness*
\begin{proof}

Now, we define $\hq_1$ and $\hq_2$ as the lowest and highest attainable 
platform payoff by a user who chooses a strategy $\hq \in \hcQ$, i.e.,
\begin{align*}
    \hq_1 = \arg\max_{\hq \in \hcQ} \brV(\pmu{\delta_{\hq}}{\cdot}, \hq),
    \qquad 
    \hq_2 = \arg\min_{\hq \in \hcQ} \brV(\pmu{\delta_{\hq}}{\cdot}, \hq).
\end{align*}
Note that for our bound to meaningful, we must have that 
\begin{align*}
    \gamma \coloneqq \frac{\brV(\pmu{\delta_{\hq_1}}{\cdot}, \hq_1)
    - \brV(\pmu{\delta_{\hq_2}}{\cdot}, \hq_2)}{4}
    - (\beta - \alpha) > 0,
\end{align*}
and so we assume that this inequality holds for the remainder of the proof.
Now, 
recall that the payoff function $V$ is scaled so that $0 \leq V(B, Z) \leq 1$. 
For some constant 
$c \in [0, 1]$
to be set later, let \[U(B, Z) = (V(B, Z) - c)^2,\]
so that for any belief $\mu$ and user strategy $q$,
\begin{align*}
    \brU(p^\mu, q, c) = \text{Var}_{(B, Z) \sim p^\mu \times q}[V(B, Z)] + \left(\brV(p^\mu, q) - c\right)^2,
\end{align*}
where we make the dependence on the unset constant $c$ explicit for notational convenience.
In turn, the function that the strategic user aims to maximize is
\begin{align*}
    \widetilde{U}(q, c) 
    \coloneqq \min_{\mu \in S^\infty(q, p, \hcQ)} \brU(p^\mu, q).
\end{align*}
When $q = \hq \in \hcQ$ and $p(\cdot;\mu)$ has full support for all $\mu$, $\tU(q, c)$ simplifies to 
\begin{align*}
    \widetilde{U}(\hq, c) 
    \coloneqq \brU(\pmu{\delta_{\hq}}{\cdot}, \hq).
\end{align*}
We now prove a few lemmata that will be useful later in the proof.
First, we show that our condition on the loss landscape of $\brV$
suffices to show a similar condition for the loss landscape of $\brU$:
\begin{lemma}
    \label{applem:perturbable}
    For any proposition distribution $r \in \Delta(\cZ)$,
    any fixed user strategy $q$,
    and any $\eps > 0$, 
    there exists a distribution $r' \in \Delta(\cZ)$ so that 
    $\mathcal{W}_1(r, r') \leq \eps$ and 
    $\delta_{r, q}(c) \coloneqq \brU(r, q, c) - \brU(r', q, c) \neq 0$
    almost everywhere (with respect to $c$).
\end{lemma}
\begin{proof}
    We can use \cref{ass:perturbable} to find a distribution $r' \in \Delta(\cZ)$ 
    such that $d(r, r') \leq \eps$ and 
    $\delta \coloneqq \brV(r', q) - \brV(r, q) \neq 0$. 
    Then, note that $\brU(r, q, c) = \brU(r', q, c)$ if and only if
    \begin{align*}
        \text{Var}_{(B, Z) \sim r \times q}[V(B, Z)] + (\brV(r, q) - c)^2
        &= 
        \text{Var}_{(B, Z) \sim r' \times q}[V(B, Z)] + (\brV(r', q) - c)^2 \\
        \text{Var}_{(B, Z) \sim r \times q}[V(B, Z)] - 
        \text{Var}_{(B, Z) \sim r' \times q}[V(B, Z)] &= 2\delta(\brV(r, q) - c) + \delta^2,
    \end{align*}
    and in particular, if and only if 
    \begin{align*}
        c = \brV(r, q) - \frac{1}{2\delta}\left(\text{Var}_{(B, Z) \sim r \times q}[V(B, Z)] - 
        \text{Var}_{(B, Z) \sim r' \times q}[V(B, Z)] - \delta^2\right),
    \end{align*}
    which is a measure-zero set.
\end{proof}
Next, we show that given a collection of proposition distributions $\{r_1, \ldots r_m\}$
corresponding to a set of user models $\{\hq_1, \ldots, \hq_m\}$ such that $r_i$ is 
$\eps_0$-close to $\pmu{\delta_{\hq_i}}{\cdot}$, we can construct a new algorithm $p'$ that is
$2\eps_0$-close to $p$ such that $p'$ satisfies the regularity conditions in \cref{ass:p-cont},
and $p'(\cdot;\delta_{\hq_i}) = r_i$ for all $i \in [m]$.
\begin{lemma}
    \label{applem:patching}
    Given an algorithm $p$ satisfying \cref{ass:p-cont},
    a collection of proposition distributions $\{r_1, \ldots r_m\}$,
    and a set of user models $\{\hq_1, \ldots, \hq_m\}$ such that 
    \[
        \mathcal{W}_1(\pmu{\delta_{\hq_i}}{\cdot}, r_i) \leq \eps_0 \ \forall\ i \in [m],
    \]
    there exists an algorithm $p'$ such that $d_\cP(p, p') \leq 2\eps_0$ and 
    \[
        p(\cdot; \delta_{\hq_i}) = r_i \forall\ i \in [m].
    \]
\end{lemma}
\begin{proof}
    By \cref{ass:p-cont}, we know that for all $Z \in \cZ$, 
    $p(Z;\mu)$ is continuous in $\mu$ around point mass beliefs.
    This continuity also implies continuity in Wasserstein distance.
    Thus, there exists a single constant $\delta$ such that for any $i \in [m]$,
    \[
        d(\delta_{\hq_i}, \mu) < \delta \implies \mathcal{W}_1(p(\cdot; \delta_{\hq}), p(\cdot; \mu)) \leq \eps_0,
    \]
    where $d(\cdot,\cdot)$ is the same distance metric with respect to which \cref{ass:p-cont} (ii) holds.
    Then, define the new algorithm 
    \[
        p'(Z;\mu) = \begin{cases}
            r_i
            &\text{if } d(\delta_{\hq_i}, \mu) < \delta \\ 
            p(Z; \mu) &\text{otherwise.}
        \end{cases}
    \]
    Clearly, this algorithm satisfies $p(\cdot; \delta_{\hq_i}) = r_i \forall\ i \in [m].$
    Furthermore, it satisfies the continuity condition since
    it is actually constant in the 
    neighborhood of each point mass belief $\delta_{\hq}$. 
    Finally,
    \begin{align*}
        d_\cP(p, p') 
        &= \sup_{\mu \in \Delta(\hcQ)} \mathcal{W}_1(p(\cdot;\mu), p'(\cdot;\mu)) \\
        &= \max_{i \in [m]}\ \sup_{d(\mu, \delta_{\hq_i}) < \delta} \mathcal{W}_1(p(\cdot;\mu), r_i) \\
        &\leq \max_{i \in [m]}\ \sup_{d(\mu, \delta_{\hq_i}) < \delta} \mathcal{W}_1(p(\cdot;\mu), p(\cdot;\delta_{\hq_i})) + \mathcal{W}_1(p(\cdot;\delta_{\hq_i}), r_i) \\
        &\leq 2\eps_0,
    \end{align*}
    concluding the proof.
\end{proof}
\noindent We now resume the main proof. For any $c$, we partition the set of user models into 
\begin{align*}
\hcQ_+(c) &= \{\hq \in \hcQ: \brV(\pmu{\delta_{\hq}}{\cdot}, \hq) > c\}, \\
\hcQ_-(c) &= \{\hq \in \hcQ: \brV(\pmu{\delta_{\hq}}{\cdot}, \hq) < c\}, \\
\hcQ_=(c) &= \{\hq \in \hcQ: \brV(\pmu{\delta_{\hq}}{\cdot}, \hq) = c\}.
\end{align*}
For any value of $c$, we define the functions
\begin{align*}
    M_+(c) = \max_{\hq \in \hcQ_+(c)} \tU(\hq, c), \qquad
    M_-(c) = \max_{\hq \in \hcQ_-(c)} \tU(\hq, c), \qquad
    M(c) = \max_{\hq \in \hcQ} \tU(\hq, c),
\end{align*}
then let \[
   Q^*(c) \coloneqq \arg\max_{\hq \in \hcQ} \tU(q, c)
\] be the set of possible strategies for a strategic user, 
assuming they play according to one of the user models 
$\hq \in \hcQ$. 

Observe that for any $\hq \in Q^*(c)$,
\begin{align}
    \nonumber
    \beta + ( \brV(\pmu{\delta_{\hq}}{\cdot}, \hq) - c)^2 &\geq \tU(\hq, c) \\ 
    \nonumber
    &\geq \max\left\{
        \tU(\hq_1), 
        \tU(\hq_2), 
    \right\} \\
    \nonumber
    &\geq \alpha + 
    \max\left\{
    \left(\max_{\hq \in \hcQ} \brV(p^{\delta_{\hq}}, \hq) - c\right)^2,
    \left(\min_{\hq \in \hcQ} \brV(p^{\delta_{\hq}}, \hq) - c\right)^2
    \right\} \\
    \nonumber
    &\geq \alpha + \frac{\lr{
        \max_{\hq \in \hcQ} \brV(p^{\delta_{\hq}}, \hq) - \min_{\hq \in \hcQ} \brV(p^{\delta_{\hq}}, \hq)
    }^2}{4} \\
    \nonumber
    (\brV(\pmu{\delta_{\hq}}{\cdot}, \hq) - c)^2 
    &\geq \frac{\lr{
        \max_{\hq \in \hcQ} \brV(p^{\delta_{\hq}}, \hq) - \min_{\hq \in \hcQ} \brV(p^{\delta_{\hq}}, \hq)
    }^2}{4} - (\beta - \alpha)\\
    &= \gamma > 0. 
    \label{eq:middle_bound}
\end{align}
In particular, this implies that $\hq \not\in \hcQ_=(c)$, and thus
\[
   M(c) = \max\{M_+(c), M_-(c)\}.
\]
Now, as $c$ increases, the set $\hcQ_-(c)$ grows and each $\tU(\hq, c)$ is non-decreasing for each $\hq \in \hcQ_-(c)$,
and so $M-_(c)$ is non-decreasing in $c$. 
Similar logic shows that $M_+(c)$ is non-increasing in $c$.
Furthermore, by definition, $M_+(c) > M_-(c)$ when 
$c = 0$
and $M_-(c) > M_+(c)$ when 
$c = 1$.
We thus let $c_0 = \inf \{c: M_-(c) \geq M(c)\},$ and let $c_1 = \sup \{c: M_+(c) \geq M(c)\}.$
\begin{lemma}
    There exists a constant $\psi > 0$ such that 
    the function $M_-(c)$ is uniformly continuous over the interval 
    $[c_0 - \psi, 1]$, 
    and $M_+(c)$ is uniformly continuous over 
    $[0, c_1 + \psi]$
\end{lemma}
\begin{proof}
    We start by considering $M_-(c)$.
    For any $\psi$, the interval 
    $[c_0 - \psi, 1]$ 
    is compact and so 
    by the Heine-Cantor theorem it suffices to show that $M_-(c)$ is pointwise continuous at 
    each $c$ in the interval. 
    Let $P \subset \bbR$ be the (finite) set 
    $\{\brV(\pmu{\delta_{\hq}}{\cdot}, \hq): \hq \in \hcQ\}$.
    Now, for any $c$ in the interval, one of the following two 
    cases must hold:
    \begin{itemize}
        \item[(A)]$c \not\in P$, i.e., there does not exist a $\hq \in \hcQ$ 
        whose corresponding platform payoff is equal to~$c$. 
        In this case, 
        we can always find a $\delta$ small enough such that for all 
        $c' \in (c - \delta, c + \delta)$, the set 
        $\hcQ_-(c') = \hcQ_-(c)$ does not change (by the finite nature of $\hcQ$). 
        We then use the continuity of
        $\tU(\hq, c)$ in $c$ for each $\hq \in \hcQ$, 
        and the fact that a maximum of continuous functions 
        is continuous to show that $M_-$ is continuous at $c$.

        \item[(B)]$c \in P$, i.e., 
        there exists at least one user model in $\hcQ$
        whose corresponding payoff is $c$ (in other words, 
        $\hcQ_=(c) \neq \emptyset$).
        Let $\hq \in \hcQ_=(c)$ be any such model.
        Define $\bar{c} = c + 2\psi$, so that $\bar{c} > c_0$. 
        We will first show by contradiction that there exists 
        a user model $\hq' \in \hcQ_-(\bar{c})$ such that
        $\tU(\hq, \bar{c}) < \tU(\hq', \bar{c})$. 
        In particular, if this were not the case,
        we would have $\tU(\hq, \bar{c}) = M_-(\bar{c})$, 
        and since $\bar{c} > c_0$ we would have that 
        $\tU(\hq, \bar{c}) \geq M(\bar{c})$,
        and thus by \eqref{eq:middle_bound},
        \begin{align*}
            (\brV(\pmu{\delta_{\hq}}{\cdot}, \hq) - \bar{c})^2 = \gamma > 0.
        \end{align*}
        By supposition, $\brV(\pmu{\delta_{\hq}}{\cdot}, \hq) = c$. 
        Thus, by setting $\psi$ small enough (i.e., as $\bar{c} \to c$), 
        we reach a contradiction. We can also set $\psi$ small enough so that
        $\hcQ_-(\bar{c}) = \hcQ_-(c) \cup \hcQ_=(c)$. In this case, the 
        logic above applies to any $\hq \in \hcQ_=(c)$, and so it must be 
        that for some $\hq' \in \hcQ_-(c)$,
        \[
            \eta \coloneqq \min_{\hq \in \hcQ_=(c)} \lr{\tU(\hq', \bar{c}) - \tU(\hq, \bar{c})} > 0,
        \]
        For the same $\hq'$ (and again, any $\hq \in \hcQ_=(c)$),
        and any $c' \in (c, \bar{c})$,
        \begin{align*}
            \tU(\hq', c') - \tU(\hq, c') 
            &\geq \eta + 
            (\tU(\hq', \bar{c}) - \tU(\hq', c'))
            - (\tU(\hq, \bar{c}) - \tU(\hq, c')).
        \end{align*}
        Observing that
        \begin{align*}
            \tU(\hq', \bar{c}) - \tU(\hq', c') 
            &= 2(\bar{c} - c')(\brV - c) + (\bar{c} - c')^2 \\
            &= 4\psi(\brV - c) + 4\psi^2,
        \end{align*}
        we can again set $\psi$ small enough so that 
        $\tU(\hq', c') > \tU(\hq, c')$. 
        As a result, we have shown that 
        for any $c' \in (c, \bar{c}),$
        the maximizer corresponding to $M_-(c')$
        is not a member of $\hcQ_=(c)$, i.e.,
        \[
            \arg\max_{\hq \in \hcQ_-(c')} \tU(\hq, c') \cap \hcQ_=(c) = \emptyset.
        \]

        Now, there exists a $\delta$ such that for all 
        $c' \in [c - \delta, c + \delta]$, 
        \[
            \hcQ_-(c') = \begin{cases}
                \hcQ_-(c) &\text{ if } c' < c \\
                \hcQ_-(c) \cap \hcQ_=(c) &\text{ if } c' > c,
            \end{cases}
        \]
        and in both cases 
        $\arg\max_{\hq \in \hcQ_-(c')} \subset \hcQ_-(c)$,
        and thus we can use continuity of 
        $\tU(\hq, c)$ in $c$.

    \end{itemize}
    The same logic implies that $M_+(c)$ is uniformly continuous on 
    the interval 
    $[0, c_1 + \psi]$.
\end{proof}
Note that by definition of $c_1$, and because $c_1 + \psi > c_1$, it must be that 
$M_+(c_1 + \psi) < M(c_1 + \psi)$ 
(otherwise $c_1$ would not be an upper bound on the set of 
which it is the $\sup$), which means that 
$M_-(c_1 + \psi) = M(c_1 + \psi)$, 
which means that $c_1 + \psi \geq c_0$. 
Conversely, $c_0 - \psi \leq c_1$, and thus 
$[c_0 - \psi, c_1 + \psi]$ is a well-defined interval on which 
both $M_-(c)$ and $M_+(c)$ are uniformly continuous.

Consider the function $h(c) \coloneqq M_+(c) - M_-(c)$ on the interval 
$[c_0 - \psi, c_1 + \psi]$. At the beginning of the interval, 
$h(c) > 0$, while at the end of the interval $h(c) < 0$.
By the intermediate value theorem, 
there exists $c^* \in [c_0 - \psi, c_1 + \psi]$ such that $h(c^*) = 0$.

Now, $\hcQ^*(c^*)$ is surely not a singleton, as it must contain at least 
one element $\hq^+ \in \hcQ_+(c^*)$ and one element $\hq^- \in \hcQ_-(c^*)$. 
We apply the following procedure to each $\hq \in \hcQ^*(c^*)$:
\begin{enumerate}
    \item Use \cref{applem:perturbable} to find a distribution 
    $r_{\hq} \in \Delta(\cZ)$
    such that $\mathcal{W}_1(r_{\hq}, \pmu{\delta_{\hq}}{\cdot}) < \eps_0$,
    and $\brU(r_{\hq}, \hq) \neq \brU(\pmu{\delta_{\hq}}{\cdot}, \hq)$.
    For simplicity, we assume that $c^*$ is not one of the finite number of
    values of $c$ such that we cannot apply \cref{applem:perturbable}---if 
    this is not the case, we can simply perturb $c^*$ by some 
    sufficiently small amount so as not to affect the calculations in the 
    rest of the proof.
    \item If $\brU(r_{\hq}, \hq) > \brU(\pmu{\delta_{\hq}}{\cdot}, \hq)$, terminate.
    \item Otherwise, continue to the next $\hq \in \hcQ^*(c^*)$.
    \item If that is the last $\hq \in \hcQ^*(c^*)$, terminate and do not perform Step 1. 
\end{enumerate}
At the end of the procedure, 
we will have some $\hq^* \in \hcQ^*(c^*)$ such that 
$\hq^*$ is strictly preferred by the user to any other 
$\hq \in \hcQ^*(c^*)$ (and thus, to any $\hq \in \hcQ$) 
under the constructed proposition distributions.
In particular, we can use \cref{applem:patching} to construct a 
new algorithm $p'$ such that
\[
    d_\cP(p, p') < \eps_0 \qquad \text{ and } \qquad 
    \arg\max_{\hq \in \hcQ} \brU(p'(\cdot;\delta_{\hq}), \hq) = \{\hq^*\}.
\]
Without loss of generality,
suppose that 
$\hq^* \in \hcQ_{-}(c^*)$.

Now, let $\eps_2 > 0$ be a small enough constant to ensure that
$Q^*(c^* - \eps_2) \subset Q_+(c^*)$ 
(note that $Q^*(c^* - \eps_2) \subset Q_+(c^* - \eps_2)$ for all $\eps_2 > 0$;
then, by setting $\eps_2$ small enough we can also ensure that $Q_+(c^* - \eps_2) = Q_+(c^*)$).
Thus, we can set $\eps_2$ small enough to ensure the following two conditions:
\begin{itemize}
    \item[(a)] $\arg\max_{\hq \in \hcQ} \brU(p(\cdot;\delta_{\hq}), \hq) \subset \hcQ_+(c^*),$ and
    \item[(b)] $\arg\max_{\hq \in \hcQ} \brU(p'(\cdot;\delta_{\hq}), \hq) \subset \hcQ_-(c^*).$
\end{itemize}
In particular, in combination with \eqref{eq:middle_bound}, these two conditions 
imply that 
\begin{itemize}
    \item[(a)] For all $\hq^*(p) \in \arg\max_{\hq \in \hcQ} \brU(p(\cdot;\delta_{\hq}), \hq)$,
    we have $\brV(p(\cdot;\delta_{\hq}), \hq) \geq c^* + \sqrt{\gamma'}$
    \item[(b)] For all $\hq^*(p') \in \arg\max_{\hq \in \hcQ} \brU(p'(\cdot;\delta_{\hq}), \hq)$,
    we have $\brV(p(\cdot;\delta_{\hq}), \hq) \leq c^* - \sqrt{\gamma'}$,
\end{itemize}
where we define 
\[
    \gamma' \coloneqq \frac{\max_{\hq \in \hcQ} \brV(p'(\cdot; \delta_{\hq}), \hq)
    - \min_{\hq \in \hcQ} \brV(p'(\cdot; \delta_{\hq}), \hq)}{4}
    - (\beta - \alpha) > 0,
\]
which can be made arbitrarily close to $\gamma$ by setting $\eps_0$ 
small enough.
Putting these two together, we get that 
\[
    \left|\brV(p(\cdot;\delta_{\hq^*(p')}), \hq^*(p'))
    - 
    \brV(p(\cdot;\delta_{\hq^*(p)}), \hq^*(p))\right|
    \geq 2\sqrt{\gamma'}.
\]

To conclude the proof, we show that for a sufficiently granular 
$\eps$-net, the user does not gain much by deviating from the 
set of user models $\hcQ$.
\begin{restatable}{lemma}{epsnet}
    \label{applem:eps_net}
    Fix any any user strategy $q$ and any full-support algorithm $p$,
    and suppose that 
    $\hcQ$ is an $\eps$-net 
    hypothesis class
    for some $\eps \in (0, \frac{1}{|\cB|})$.
    Then for every user model $\hq \in \hcQ$,
    \begin{align*}
        \hq \in S^\infty(q, p, \hcQ) \implies
        \max_{Z \in \cZ} \KL(q(\cdot|Z), \hq(\cdot|Z)) \leq \log\left(
            \frac{1}{1-|\cB|\cdot \eps}
        \right),
    \end{align*}
    and as a result, $\tU(\hq, c) \geq \tU(q, c) - \sqrt{-\frac{1}{2}\log(1-|\cB|\eps)}$.
\end{restatable}
\begin{proof}
    Note that because $\hcQ$ is the Cartesian product of $\Delta_\eps(\cB)$ 
    across $Z$, when $p$ has full support 
    every element $\hq \in S^\infty(q, p, \hcQ)$ must satisfy 
    \[
        \hq(\cdot|Z) \in \arg\min_{p_B \in \Delta_\eps(\cB)} \KL(q(\cdot|Z), p_B)\ \forall\ Z \in \cZ.
    \] 
    In particular, we can find a $\hq'$ that strictly dominates 
    any $\hq$ that violates this condition by 
    swapping its behavior at violating values of $Z$.
    Now, for 
    $\nu = \eps \cdot |\cB|$,
    let $q_\nu$ be a mixture of 
    $q$ with the uniform distribution over $\cB$, i.e.,
    \[
        q_\nu(\cdot|Z) \coloneqq \nu \cdot \frac{1}{|\cB|} + (1 - \nu) \cdot q(\cdot|Z).
    \]
    By definition of the $\eps$-net, there must be some $\hq_i$ such that 
    $\hq_i(B|Z) \geq q_\nu(B|Z) - \eps$, and in turn,
    \begin{align*}
        \KL(q(\cdot|Z), \hq_i(\cdot|Z)) 
        &\leq \mathbb{E}_{B \sim q(\cdot|Z)} \left[
            \log\lr{
                \frac{q(\cdot|Z)}{\nu \cdot \frac{1}{|\cB|} + (1 - \nu) \cdot q(\cdot|Z) - \eps}
            }
        \right] 
    =
        \log\lr{ \frac{1}{1 - \nu} }.
    \end{align*}
    Next, observe that the function 
    $U(B, Z) = (V(B, Z) - c)^2$ is 
    
    bounded in $[0, 1]$, and so
    \begin{align*}
        \tU(q, c) 
        &= \min_{\mu \in S^\infty(q, p, \hcQ)} \brU(\pmu{\mu}{\cdot}, q) \\
        &\leq \brU(\pmu{\delta_{\hq}}{\cdot}, q) &\text{(since $\hq \in S^\infty$)} \\
        &\leq \brU(\pmu{\delta_{\hq}}{\cdot}, \hq) + \max_{Z \in \cZ}\ \text{TV}(\hq(\cdot|Z), q(\cdot|Z)) &\text{(since $U$ is bounded in $[0, 1]$)}\\
        &\leq \tU(\hq, c) + \sqrt{\frac{1}{2}\log\lr{\frac{1}{1-\nu}}},
    \end{align*}
    where above we used the definition of total variation distance,
    as well as the fact that for any two probability distributions $A$ and $B$,
    \(
        \text{TV}(A, B) < \sqrt{\frac{1}{2}\KL(A, B)}
    \)
    by Pinsker's inequality.
\end{proof}
\noindent For any $\eps_1 > 0$, we can use \cref{applem:eps_net} and set
\[
    \eps = \frac{1 - \exp\lr{
        -\eps_1^2
    }}{|\cB|}
\]
to get that $\tU(\hq, c) \geq \tU(q, c) - \eps_1$.
Thus, if 
$\hq \in S^\infty(q, p, \hcQ)$ for some strategy $q$,
then $q$ 
cannot yield a significantly higher payoff than $\hq$.
This allows us to reduce the case of picking the 
optimal user strategy $q$ to picking the optimal user model $\hq \in \hcQ$.
That is, if the user {\em strictly} prefers a user model $\hq^*$ to any 
other user model, 
we can set $\eps$ sufficiently small so that the user strictly
prefers their globally stable set to be $\{\hq^*\}$,
which entails playing a strategy close to $\hq^*$.

Thus, the smallest possible gap between the platform's predicted payoff under 
$p'$ and its true payoff under $p'$ is given by
\begin{align*}
    \min_{\mu \in \Delta(\{\hq^*(p)\})}
    \left|
        \widehat{V}(p', \mu) - \brV^*(p')
    \right| 
    &= 
    \left|
        \widehat{V}(p', \delta_{\hq^*(p)}) - \brV(p'(\cdot; \delta_{\hq^*(p')}, q^*(p')))
    \right|  \\ 
    &= \left|
        \brV(p'(\cdot; \delta_{\hq^*(p)}), \hq^*(p)) - \brV(p'(\cdot; \delta_{\hq^*(p')}, q^*(p')))
    \right|  \\ 
    &\geq 2\sqrt{\gamma'} - \left|
        \brV(p'(\cdot; \delta_{\hq^*(p')}, \hq^*(p'))) - \brV(p'(\cdot; \delta_{\hq^*(p')}, q^*(p')))
    \right|  \\ 
    &\geq 2\sqrt{\gamma'} - \eps_1,
\end{align*} 
where the last inequality follows from \cref{applem:eps_net}.

\end{proof}

\clearpage

\subsection{Proof of \cref{prop:cf_omega}}
\label{app:cf_omega_pf}

We prove this result using a highly oversimplified example in order to illustrate the
main principle behind the proof.
First, let the user's payoff function be $U(B, Z) = V(B, Z) + \lambda \cdot g(Z)$,
where $V(B, Z)$ is the platform payoff and $g(Z)$ is a function 
to be specified later (along with the scalar $\lambda$). 
Even without specifying $g(Z)$,
it is clear that for any $Z \in \cZ$,
$q^\BR(\cdot|Z) = \arg\max_{B \in \cB} V(B, Z)$.

Now, find the following two user models $\hq_1$ and $\hq_2$:
\begin{align*}
    \hq_1 &= \arg\min_{\hq \in \hcQ} \KL\left[p(\cdot; \delta_{\hq}) \times q^\BR, p(\cdot; \delta_{\hq}) \times \hq\right], \\
    \hq_2 &= \arg\max_{\hq \in \hcQ} \mathcal{W}\left[p(\cdot; \delta_{\hq_1}), p(\cdot; \delta_{\hq})\right].
\end{align*}
(In the literature, $\hq_1$ is referred to as a Berk-Nash equilibrium \citep{esponda2016berk,frick2020stability}.)
Consider the distributions
\( p(\cdot;\delta_{\hq_1}) \) 
and 
\( p(\cdot;\delta_{\hq_2}) \), and define 
\[
    g(\cdot) \coloneqq \arg\max_{\|f\|_L \leq 1}
    \mathbb{E}_{Z \sim p(\cdot;\delta_{\hq_2})}[f(Z)] -
    \mathbb{E}_{Z \sim p(\cdot;\delta_{\hq_1})}[f(Z)],
\]
where $\|f\|_L$ represents the Lipschitz constant of the function $f: \cZ \to
\bbR$.

By construction of $g$, we have that for any user strategy $q$,
\begin{align*}
    \brU(p(\cdot; \delta_{\hq_2}), q) - \brU(p(\cdot; \delta_{\hq_1}), q) 
    &\geq \lambda \cdot \mathcal{W}(\pmu{\delta_{\hq_1}}{\cdot}, \pmu{\delta_{\hq_2}}{\cdot}) - 1 \\ 
    &\geq \frac{\lambda}{2}\max_{\hq, \hq' \in \hcQ} \mathcal{W}(\pmu{\delta_{\hq}}{\cdot}, \pmu{\delta_{\hq'}}{\cdot}) - 1 \\
    &> 0,
\end{align*}
as long as $\lambda > \frac{2}{R}$.

Now, let $\hcQ_1 = \{\hq_2\}$ and $\hcQ_2 = \{\hq_2, \hq_1\}$.
Under $\hcQ_1$, the platform will trivally converge to $\hq_2$
regardless of the user's behavior, and so the user is incentivized
to play according to $q^\BR$. 

When $\hq_1$ is added to the set of user models, 
the naive user strategy $q^\BR$ will lead to the user model 
$\hq_1$ never being eliminated from the globally stable set 
(due to its status as a Berk-Nash equilibrium for $q^\BR$).
By construction of the user payoff function, this is always
suboptimal for the user, since $\delta_{\hq_1}$ will be in 
$\Delta(S^\infty(q, p, \hcQ))$, 
and so they will be incentivized to switch to a strategy that 
ensures only $\hq_2$ is in the globally stable set.

Since, by assumption, $V(B, Z)$ has a unique maximizer for each $Z \in \cZ$,
the strategic user's new strategy must result in strictly lower platform 
payoff, concluding the proof of the theorem.

\begin{remark}
    Note that while the example presented in this proof is rather contrived, 
    the principle behind it is actually quite general.
    The principle is that the strategic user wants to induce a specific 
    proposition distribution from the platform, but does not want to stray 
    too far from their best-response behavior. 
    Thus, they purposefully pick a strategy that the platform misinterprets (in
    this case, that strategy is just $q^\BR$, but it could be any strategy $q_1^*$).
    When the platform gets ``better'' at capturing their behavior by 
    adding a user model that is close to $q_1^*$, the user is forced to 
    move even further away from their best-response behavior, making things 
    worse for the platform.
\end{remark}

\clearpage
\subsection{Proof of \cref{prop:smooth_p_BR}}
\label{app:best_response_predictable_proof}
\bestresponsepredictable*
\begin{proof}
    We first restate the two quantities being compared:
    \begin{align*}
        \hV(\pcf, \mu) &= \bbE_{\hq \sim \mu}\left[
            \brV(\pcf(\cdot; \mu), \hq)
        \right] .
        \\
        \brV_\BR ( 
        \pcf
        )
        &= 
        \min_{\mu \in \Delta(S^\infty (  q^\BR , \pcf))}
        \brV 
        \left( 
        \pcf(\cdot; \mu) , 
        q^\BR
        \right) .
    \end{align*}
    Using the triangle inequality,
    \begin{align*}
		\max_{\mu \in \Delta(S^\infty(q^\BR, p))} \left|
		\hV ( \pcf, \mu)
		- 
		\brV_\BR ( 
		\pcf
		)
		\right| 
        &\leq\!\!\!\!\!
        \max_{\mu_1, \mu_2 \in \Delta(S^\infty(q^\BR, p))}\! \left|
            \bbE_{\hq \sim \mu_1}\!\left[
                \brV(\pcf(\cdot; \mu_1), \hq)
            \right]
            \!-\! 
            \brV\!
            \left( 
            \pcf(\cdot; \mu_2) , 
            q^\BR
            \right) 
        \right|  \\
        &\leq  \max_{\mu_1, \mu_2 \in \Delta(S^\infty(q^\BR, p))}\! \left|
            \bbE_{\hq \sim \mu_1}\!\left[
                \brV(\pcf(\cdot; \mu_1), \hq)
            \!-\! 
            \brV\!
            \left( 
            \pcf(\cdot; \mu_2) , 
            \hq
            \right) 
            \right]
        \right|  
        \\
        &\qquad + \left|
            \bbE_{\hq \sim \mu_1}\!\left[
                \brV(\pcf(\cdot; \mu_2), \hq)
            \!-\! 
            \brV\!
            \left( 
            \pcf(\cdot; \mu_2) , 
            q^\BR
            \right) 
            \right]
        \right|   \\
        &\leq  \max_{\mu_1, \mu_2 \in \Delta(S^\infty(q^\BR, p))}\! \left(\left|
            \bbE_{\hq \sim \mu_1}\!\left[
                \brV(\pcf(\cdot; \mu_1), \hq)
            \!-\! 
            \brV\!
            \left( 
            \pcf(\cdot; \mu_2) , 
            \hq
            \right) 
            \right] \!\!\!\!\!\!\!\!\!\!\!\!\!\phantom{q^{\BR^2}_2}
        \right| \right.
        \\
        &+ \left.\max_{\hq \in S^\infty(q^\BR, p)} \left|
                \brV(\pcf(\cdot; \mu_2), \hq)
            \!-\! 
            \brV\!
            \left( 
            \pcf(\cdot; \mu_2) , 
            q^\BR
            \right) 
        \right|  \right) 
    \end{align*}
    Note that the second term above is bounded by
    the maximum total variation distance between $\hq$ and 
    $q^\BR$, which we can bound using \cref{applem:eps_net}, restated below:
    \epsnet*
    \noindent We can then bound the first term using the well-behavedness condition 
    which implies that
    \begin{align*}
        d_{\cP}(\pcf(\cdot; \mu_1), \pcf(\cdot; \mu_2)) 
        &\leq L_{\cP} \cdot \mathbb{E}_{\hq_1 \sim \mu_1,\,\hq_2 \sim \mu_2}\left[
            \max_{Z \in \cZ}\ 
            \text{TV}(\hq_1(\cdot|Z), \hq_2(\cdot|Z))
            \right] \\
        \leq L_{\cP} \cdot &\mathbb{E}_{\hq_1 \sim \mu_1,\,\hq_2 \sim \mu_2}\left[
            \max_{Z \in \cZ}\ 
            \text{TV}(\hq_1(\cdot|Z), q^\BR(\cdot|Z))
            + \text{TV}(\hq_2(\cdot|Z), q^\BR(\cdot|Z))
            \right],
    \end{align*}
    which we can again bound using \cref{applem:eps_net}.
    Thus,
    \begin{align*}
        \max_{\mu \in \Delta(S^\infty(q^\BR, p))} \left|
            \hV ( \pcf, \mu)
            - 
            \brV_\BR ( 
            \pcf
            )
        \right| 
        \leq (2\cdot L_{\cP}+1) \sqrt{\frac{1}{2}\log\left(\frac{1}{1-|\cB|\cdot \eps}\right)}
        \leq (2L_{\cP} + 1) \cdot \sqrt{|\cB|\cdot\eps},
    \end{align*} 
    where the last inequality follows from $-\log(1-x) \leq 2x$ for all $x \in [0, \frac{1}{2}]$.
\end{proof}

\end{document}